%% file: Fast_Quantization_of_SV_Models.tex
\documentclass[11pt,a4paper,notitlepage]{article}
\usepackage{amsmath}
\usepackage{amsthm}
\usepackage{amssymb}
\usepackage{mathrsfs}
\usepackage{lipsum}
\usepackage[square]{natbib}
\usepackage{graphicx}
\usepackage{hyperref}
\hypersetup{
    colorlinks=true,  
    linkcolor=blue,   
    citecolor=blue
}
\usepackage{authblk}
\usepackage{bbm}
\usepackage{color}
\usepackage{pgf}
\usepackage{a4wide}	
\usepackage{mathtools}
\usepackage{titlesec}

\newcommand{\R}{\mathbb{R}}                         

\newcommand{\E}[1]{\mathbb{E}\left[#1\right]}       
\renewcommand{\P}{\mathbb{P}}                       
\newcommand{\sigalg}[1]{\mathscr{#1}}               
\newcommand{\F}{\sigalg{F}}                         
\newcommand{\triple}{(\Omega,\F,\P)}                
\newcommand{\ind}[1]{\mathbb{I}_{\left\{#1\right\}}}

\newcommand{\tr}{t\in[0,T]}                         
\newcommand{\filt}{(\F_t)_{\tr}}                    
\newcommand{\basis}{(\Omega,\F,\filt,\P)}           

\renewcommand{\rm}[1]{r^{#1-}}          
\newcommand{\Rm}[2]{r_{#1}^{#2-}}
\newcommand{\Rp}[2]{r_{#1}^{#2+}}
\newcommand{\Rpm}[2]{r_{#1}^{#2\pm}}

\newcommand{\affm}[2]{m_{#1}^{#2}}
\newcommand{\affc}[2]{c_{#1}^{#2}}

\newcommand{\Xq}{\widehat{X}}           
\newcommand{\Xd}{\widetilde{X}}         
\newcommand{\Yd}{\widetilde{Y}}			
\newcommand{\Yq}{\widehat{Y}}           
\newcommand{\Delt}{\Delta t}            

\newcommand{\U}{\mathcal{U}}
\newcommand{\Ubar}{\overline{\mathcal{U}}} 
\newcommand{\mbar}{\overline{m}}
\newcommand{\cbar}{\overline{c}}

\newcommand{\mP}{\mathbf{P}}
\newcommand{\mJnt}{\mathbf{J}}
\newcommand{\mM}{\mathbf{M}}
\newcommand{\mf}{\mathbf{f}}
\newcommand{\mm}{\mathbf{m}}
\newcommand{\mc}{\mathbf{c}}
\newcommand{\mJ}{\mathbf{1}}

\renewcommand{\mp}{\mathbf{p}}
\renewcommand{\mho}{\mathbf{h}_{\mathsf{off}}}
\newcommand{\mhm}{\mathbf{h}_{\mathsf{main}}}
\newcommand{\hprod}{\circ}
\newcommand{\mGamma}{\mathbf{\Gamma}}

\providecommand{\abs}[1]{\lvert#1\rvert}

\newcommand{\sgn}{\,\mathrm{sgn}}

\theoremstyle{plain}
\newtheorem{thm}{Theorem}[section]
\newtheorem{prop}[thm]{Proposition}

\newtheorem{rem}[thm]{Remark}

\title{Fast Quantization of Stochastic Volatility Models}
\author[1]{Ralph Rudd}
\author[1,2]{Thomas A. McWalter\thanks{Correspondence: tom@analytical.co.za}}
\author[1,3]{J\"{o}rg Kienitz}
\author[1,4]{Eckhard Platen}
\affil[1]{Department of Actuarial Science and the African Collaboration for Quantitative Finance and Risk Research, University of Cape Town}
\affil[2]{Department of Finance and Investment Management, University of Johannesburg}
\affil[3]{Fachbereich Mathematik und Naturwissenschaften,
Bergische Universit\"{a}t Wuppertal}
\affil[4]{Finance Discipline Group and School of Mathematical and Physical Sciences, University of Technology Sydney}

\begin{document}

\maketitle

\begin{abstract}
Recursive Marginal Quantization (RMQ) allows fast approximation of solutions to stochastic differential equations in one-dimension. When applied to two factor models, RMQ is inefficient due to the fact that the optimization problem is usually performed using stochastic methods, e.g., Lloyd's algorithm or Competitive Learning Vector Quantization. In this paper, a new algorithm is proposed that allows RMQ to be applied to two-factor stochastic volatility models, which retains the efficiency of gradient-descent techniques. By margining over potential realizations of the volatility process, a significant decrease in computational effort is achieved when compared to current quantization methods. Additionally, techniques for modelling the correct zero-boundary behaviour are used to allow the new algorithm to be applied to cases where the previous methods would fail. The proposed technique is illustrated for European options on the Heston and Stein-Stein models, while a more thorough application is considered in the case of the popular SABR model, where various exotic options are also priced.
\end{abstract}

\section{Introduction}
\label{Sec: Introduction}
\input{Sections/Introduction.tex}

\section{Quantization of Single-factor Models}
\label{Sec: Single-factor}
\input{Sections/SingleFactor}

\section{Quantization of Stochastic Volatility Models}
\label{Sec: Stochastic Volatility}
\input{Sections/StochasticVolatility.tex}

\section{Computing the Joint Probabilities}
\label{Sec: Joint Probability}
\input{Sections/JointProbability.tex}
\section{Implementing the Algorithm}
\label{Sec: Implementation}
\input{Sections/Implementation.tex}

\section{Pricing European Options}
\label{Sec: Pricing European Options}
\input{Sections/Europeans.tex}

\section{Pricing Exotic Options}
\label{Sec: Pricing Exotic Options}
\input{Sections/Exotics.tex}

\section{Conclusion}
\label{Sec: Conclusion}
\input{Sections/Conclusion.tex}

\clearpage

\bibliographystyle{abbrvnat}
\bibliography{RMQ_References}

\newpage
\appendix
\input{Sections/Appendix}

\end{document}

%% file: Sections/Introduction.tex
Quantization is a lossy compression technique that has been applied to many challenging problems in mathematical finance, including pricing options with path dependence and early exercise \citep{PagesWilbertz2012, sagna2010pricing, bormetti2016backward}, stochastic control problems \citep{PagesPhamPrintems2004} and non-linear filtering \citep{PagesPham2005}.

\cite{pages2015recursive} introduced a technique known as Recursive Marginal Quantization (RMQ), which approximates the marginal distribution of a stochastic differential equation by recursively quantizing the Euler approximation of the process. This was extended to higher-order schemes by \cite{mcwalter2017recursive}. RMQ can be applied to any one-dimensional SDE, even when the transition density is unknown, and has been used to efficiently calibrate a local volatility model by \cite{callegaro2014pricing,Callegaroetal2015a}.

Applying the standard RMQ technique to a two-factor SDE generally requires the use of stochastic numerical methods, such as the randomized Lloyd's method or stochastic gradient descent methods such as Competitive Learning Vector Quantization (see \cite{pagesintroduction} for an overview of these methods). The computational cost of these techniques is prohibitive.

To overcome this numerical inefficiency, \cite{callegaro2015pricing} used conditioning to derive a modified RMQ algorithm that can be applied to stochastic volatility models while retaining the use of the underlying Newton-Raphson technique. This was achieved by performing a one-dimensional RMQ on the volatility process and then conditioning on the realizations of the resultant quantized process. We derive a new RMQ algorithm for the stochastic volatility setting by showing that the correlation between the two processes may be neglected when minimizing the distortion. We call this innovation the Joint Recursive Marginal Quantization (JRMQ) algorithm. It results in an increase in accuracy and a large increase in efficiency. Furthermore, it allows for the modelling of the correct zero-boundary behaviour of the underlying processes. We now provide an overview of the paper.

In Section \ref{Sec: Single-factor} an overview of the RMQ algorithm in the one-dimensional case is provided. Section \ref{Sec: Stochastic Volatility} derives the JRMQ algorithm for the stochastic volatility setting with the main result of the paper contained in Proposition \ref{Thm: Dimension Reduction}. Section \ref{Sec: Joint Probability} discusses how to efficiently compute the joint probabilities required by the new algorithm.
In Section \ref{Sec: Implementation}, a concise matrix formulation is provided to ease implementation. Section \ref{Sec: Pricing European Options} prices European options under the Stein-Stein, Heston and SABR stochastic volatility models. In Section \ref{Sec: Pricing Exotic Options}, a single grid generated by the JRMQ algorithm for the SABR model is used to price Bermudan and barrier options, and volatility corridor swaps. Section \ref{Sec: Conclusion} concludes.

%% file: Sections/SingleFactor.tex
Let $X$ be a continuous random variable, taking values in $\R$, and defined on the probability space $\triple$. We seek an approximation of this random variable, denoted $\Xq$, taking values in a set of finite cardinality, $\Gamma^x$, with the minimum expected squared Euclidean difference from the original. Constructing this approximation is known as vector quantization, with $\Xq$ called the \emph{quantized} version of $X$ and the set $\Gamma^x = \{x^1, \dots, x^{N}\}$ known as the \emph{quantizer} with cardinality $N$. The elements of $\Gamma^x$ are called \emph{codewords}.

The primary utility of quantization is the efficient approximation of expectations of functionals of the random variable $X$ using
\[
    \E{H(X)}=\int_{\R}H(x)\,d\P(X\leq x)\approx\sum_{i=1}^N H(x^i)\P\big(\widehat{X}=x^i\big), 
\]
where $\Xq$ denotes the quantized version of $X$.
We now briefly describe the mathematics of vector quantization. Consider the nearest-neighbour projection operator, $\pi_{\Gamma^x}:\R\mapsto\Gamma^x$, given by
\begin{align*}
    \pi_{\Gamma^x}(X):=\big\{x^{i}\in\Gamma^x\,\big|\,\,\|X-x^{i}\|\leq \|X-x^{j}\|\text{ for all $j=1,\ldots,N;$ }&\text{where equality}\\
    &\text{holds only for $i<j$}\big\}.
\end{align*}
The quantized version of $X$ is defined in terms of this projection operator as $\Xq := \pi_{\Gamma^x}(X)$. The \emph{region} $R^i(\Gamma^x)$, for $1\leq i \leq N$, is defined as
\[
    R^i(\Gamma^x):=\big\{z\in\R\,\big|\,\,\pi_{\Gamma^x}(z)=x^{i}\big\},
\]
and is the subset of $\R$ mapped to codeword $x^i$ through the projection operator.

The expected squared Euclidean error, known as the \emph{distortion}, is given by
\begin{align*}
    D(\Gamma^x)&=\E{\|X-\widehat{X}\|^2}\notag\\
    &=\int_{\R}\|x-\pi_{\Gamma^x}(x)\|^2\,d\P(X\leq x)\notag\\
    &=\sum_{i=1}^N\int_{R_i(\Gamma^x)}\|x-x^{i}\|^2\,d\P(X\leq x),
\end{align*}
and is the function that must be minimized in order to obtain the optimal quantizer.
We retain the symbol $x$ to refer to the continuous domain of the distribution of the random variable $X$, whereas $x^i$ refers to the discrete codewords of the resulting quantizer, $\Gamma^x$, for $1\leq i \leq N$.

When the gradient and Hessian of the distortion can be computed in closed-form, a simple Newton-Raphson algorithm may be used to minimise the distortion,
\[
    \prescript{(n+1)}{}{\mGamma}^x=\prescript{(n)}{}{\mGamma}^x-\left[\nabla^2D\left(\prescript{(n)}{}{\mGamma}^x\right)\right]^{-1}\nabla D\left(\prescript{(n)}{}{\mGamma}^x\right).
\]
Here, $0\leq n<n_{\mathrm{max}}$ is the iteration index of the algorithm and $[ \prescript{(n)}{}{\mGamma}^x]_i = x^i$, for $1\leq i \leq N$, is a column vector containing the codewords. The gradient vector and Hessian matrix of the distortion are $\nabla D\left(\prescript{(n)}{}{\mGamma}^x\right)$ and $\nabla^2D\left(\prescript{(n)}{}{\mGamma}^x\right)$, respectively. Note that the distortion function is applied element-wise to the column vector $\prescript{(n)}{}{\mGamma}^x$, and $\prescript{(0)}{}{\mGamma}^x$ is an initial guess for the quantizer. \cite{mcwalter2017recursive} provide explicit expressions for the gradient vector and Hessian matrix in the one-dimensional case, and an efficient matrix formulation for implementation.

To extend the applicability of vector quantization for use with SDEs, \cite{pages2015recursive} proposed recursive marginal quantization of the Euler scheme for an SDE. In order to fix the notation used in the remainder of the paper we briefly specify this problem.

Consider the one-dimensional continuous-time stochastic differential equation
\begin{align}
dX_t &= a^x(X_t)\,dt + b^x(X_t)\, dW^x_t, \qquad X_0 = x_0,\notag
\intertext{defined on $\basis$, a filtered probability space satisfying the usual conditions. The discrete-time Euler approximation $\Xd$ of $X$, on an evenly spaced time grid, is given by}
\Xd_{k+1}&=\Xd_{k}+a^x(\Xd_k)\Delt +b^x(\Xd_k)\sqrt{\Delt}Z^x_{k+1} \notag \\
&=:\U^x(\Xd_{k},Z^x_{k+1}), \label{Eq: X-process Euler Approximation}
\end{align}
for $0\leq k<K$, where $\Delt = T/K$, and $Z^x_{k+1}\sim\mathscr{N}(0,1)$ are independent standard Gaussian random variables.

The optimal quantizer for the continuous-time process $X$, at each fixed time ${t_{k+1}} = (k+1)\Delt$, should be computed using the distortion
\[ \E{\|X_{t_{k+1}}-\pi_{\Gamma^x}(X_{t_{k+1}})\|^2}. \]
This is, however, not possible in the general case, since the distribution of $X_{t_{k+1}}$ is unknown. We instead consider the distortion computed in terms of the Euler approximation $\Xd_{k+1}$.

Let $\Gamma^x_{k}$ be the quantizer of $\Xd$, at time-step $k$ with $0\leq k\leq K$. To remain consistent with the specification of the problem above, the quantizer at the initial time, $t_0$, is given by $\Gamma^x_0=\{x_0\}$. We fix the cardinality of the quantizers at all other time steps to be $N^x$ --- this may, however, be relaxed (see, for example, the discussion on optimal dispatching in \cite{pages2015recursive}). Since the Euler update is normally distributed, the quantizer at the first time step is just the vector quantization of a normal distribution.  The distortion of the quantizer for each successive step is then given by
\begin{align*}
    \widetilde{D}\big(\Gamma^x_{k+1}\big)&=\E{\big\|\Xd_{k+1}-\pi_{\Gamma^x}(\Xd_{k+1})\big\|^2}\notag\\
    &=\E{\E{\left.\big\|\Xd_{k+1}-\Xq_{k+1}\big\|^2\,\right|\,\Xd_k}}\notag\\
    &=\E{\E{\left.\big\|\U^x(\Xd_k,Z^x_{k+1})-\Xq_{k+1}\big\|^2\,\right|\,\Xd_k}}\notag\\
    &=\int_{\R}\E{\big\|\U^x(x,Z^x_{k+1})-\Xq_{k+1}\big\|^2}\!d\P(\Xd_k\leq x).
\end{align*}
To proceed, we approximate the above distortion using the distribution of $\Xq_k$ rather than $\Xd_k$, in which case
\[
    \widetilde{D}\big(\Gamma^x_{k+1}\big)\approx D\big(\Gamma^x_{k+1}\big):=\sum_{i=1}^{N^x}\E{\big\|\U^x(x_k^i,Z_{k+1})-\Xq_{k+1}\big\|^2}\P\big(\Xq_k=x_{k}^{i}\big),
\]
where the approximate distortion is defined without any accents.

This is the one dimensional vector quantization problem, where the distribution being quantized is a marginal distribution consisting of the probability-weighted sum of Euler updates, each having originated from the codewords in the quantizer at the previous time step. Recursively applying this  procedure to the updates of the Euler process is known as recursive marginal quantization (RMQ). Since the vector quantization problem specified in this manner is one-dimensional, the efficient Newton-Raphson procedure can be used to minimize the resulting distortion, which yields the quantizer at each time-step.
\cite{mcwalter2017recursive} derive explicit and efficient expressions for the gradient vector and Hessian matrix required for the Newton-Raphson procedure and show that recursive marginal quantization of higher-order schemes is possible --- specifically the Milstein and simplified weak-order 2.0 schemes. In the present work, we only consider the Euler scheme.

Note that the Euler update \eqref{Eq: X-process Euler Approximation} can be written in an affine form as
\begin{equation}
\U^x(\Xd_{k},Z^x_{k+1}) = \affm{k}{i} Z^x_{k+1} + \affc{k}{i}  \label{Eq: Euler X Affine Form},
\end{equation}
with
\begin{equation}
    \affm{k}{i}:=b^x(x_{k}^{i})\sqrt{\Delt}\qquad\text{and}\qquad
    \affc{k}{i}:=x_{k}^{i}+a^x(x_{k}^{i})\Delt.
\end{equation}
Thus, for a given quantizer, $\Gamma^x_{k+1}$, the standardized region boundaries associated with each codeword are given by
\begin{equation}
    r_{k+1}^{i,j\pm}=\frac{\tfrac{1}{2}(x_{k+1}^{j\pm1}+x_{k+1}^{j})-\affc{k}{i}}{\affm{k}{i}}, \label{Eq: X Region Boundaries}
\end{equation}
for $1\leq i, j\leq N^x$. This refers to the upper and lower region boundary of codeword $x_{k+1}^j$ when viewed from codeword $x_k^i$.
Equations \eqref{Eq: Euler X Affine Form} to \eqref{Eq: X Region Boundaries} are central to the standard RMQ algorithm, see \cite{mcwalter2017recursive}, and are presented here for use later in the paper.

%% file: Sections/StochasticVolatility.tex
In this section, we consider the recursive marginal quantization of a generic stochastic volatility model described by the coupled SDEs
\begin{align}
dX_t &= a^x(X_t)\,dt + b^x(X_t)\, dW^x_t, & X_0 &= x_0, \label{Eq: X-process}\\
dY_t &= a^y(Y_t)\,dt + b^y(X_t,Y_t)\,(\rho\,dW^x_t + \sqrt{1 - \rho^2}\,dW^\perp_t), & Y_0 &= y_0 \label{Eq: Y-process}
\end{align}
defined on $\basis$, where $W^x_t$ and $W^\perp_t$ are independent standard Brownian motions. In this system, the Cholesky decomposition, specified in terms of the correlation parameter $\rho\in[-1,1]$, is chosen explicitly in order to facilitate derivations. Here, $X_t$, referred to as the \emph{independent} process, drives the specification of the stochastic volatility factor in the \emph{dependent} process $Y_t$.

The Euler scheme for the above system is given by
\begin{align}
\Xd_{k+1} &= \U^x(\Xd_k, Z^x_{k+1}), & \Xd_0 &= x_0 \label{Eq: X-Euler}\\
\Yd_{k+1} &= \Yd_k + a^y(\Yd_k) + b^y(\Xd_k,\Yd_k) \sqrt{\Delt}(\rho Z^x_{k+1} + \sqrt{1 - \rho^2} Z^\perp_{k+1}), & \Yd_0 &= y_0 \label{Eq: Y-Euler} \\
&=: \U^y(\Xd_k, \Yd_k, Z^x_{k+1}, Z^\perp_{k+1}),\label{Eq: Y-update}
\end{align}
for $0\leq k<K$, where $\U^x(\Xd_k, Z^x_{k+1})$ is defined by \eqref{Eq: X-process Euler Approximation} and $Z^x_{k+1},Z^\perp_{k+1}\sim\mathscr{N}(0,1)$ are independent standard Gaussian random variables. The main result of this paper is to show that quantizing the Euler update  $\Yd_{k+1}=\U^y(\Xd_k, \Yd_k, Z^x_{k+1}, Z^\perp_{k+1})$ is equivalent to quantizing the update given by
\begin{equation}
    \overline{\U}^y(\Xd_k,\Yd_k, Z)=\Yd_k + a^y(\Yd_k)\Delt + b^y(\Xd_k,\Yd_k)\sqrt{\Delt} Z, \label{Eq: Margined Update}
\end{equation}
where $Z\sim\mathscr{N}(0,1)$ is any standard Gaussian random variable. Having established this result, we proceed to quantize the system and derive a one-dimensional vector quantization algorithm based on a Newton-Raphson iteration.

\begin{prop}
\label{Thm: Dimension Reduction}
Given the Euler scheme defined by \eqref{Eq: X-Euler} and \eqref{Eq: Y-Euler}, the distortion of the quantizer $\Gamma^y_{k+1}$ may be expressed as
\[
    \widetilde{D}(\Gamma^y_{k+1})=\E{\big\|\overline{\U}^y(\Xd_k,\Yd_k,Z)-\Yq_{k+1}\big\|^2},
\]
where the margined update function is defined by \eqref{Eq: Margined Update} with $Z\sim\mathscr{N}(0,1)$.
\end{prop}
\begin{proof}
The distortion of the quantizer $\Gamma^y_{k+1}$ for $\Yd_{k+1}$ is given in terms of the update \eqref{Eq: Y-update} as
\begin{align*}
    \widetilde{D}\big(\Gamma_{k+1}^y\big):=&\E{\big\|\Yd_{k+1}-\Yq_{k+1}\big\|^2}\\
    =&\E{\E{\left.\big\|\Yd_{k+1}-\Yq_{k+1}\big\|^2\,\right|\Xd_k,\Yd_k}}\\
    =&\int_{\R^2}\E{\big\|\U^y(x,y,Z^x_{k+1},Z^\perp_{k+1})-\Yq_{k+1}\big\|^2}\,d\P(\Xd_k\leq x,\Yd_k\leq y)\\
    =&\int_{\R^2}\E{f\big(\U^y(x,y,Z^x_{k+1},Z^\perp_{k+1})\big)}\,d\P(\Xd_k\leq x,\Yd_k\leq y),
\end{align*}
where $f(w) := \left(w - \pi_{\Gamma^y_{k+1}}(w)\right)^2$. The inner expectation may be written explicitly as
\[
    \E{f\big(\U^y(x,y,Z^x_{k+1},Z^\perp_{k+1})\big)}=\frac{1}{2\pi}\int_{\R^2}f\big(\U^y(x,y,u,v)\big)
    \exp\left(\frac{-u^2}{2}\right)\exp\left(\frac{-v^2}{2}\right)\,dv\,du.
\]
Now, let
\[
    z=\rho u+\sqrt{1-\rho^2}v,
\]
which means that
\[
    v=\frac{z-\rho u}{\sqrt{1-\rho^2}}\qquad\text{and}\qquad dv=\frac{1}{\sqrt{1-\rho^2}}\,dz.
\]
Then,
\begin{align*}
    &\E{f\big(\U^y(x,y,Z^x_{k+1},Z^\perp_{k+1})\big)}\\
    &\qquad=\frac{1}{2\pi\sqrt{1-\rho^2}}\int_{\R^2}f\big(\Ubar^y(x,y,z)\big)
    \exp\left(\frac{-u^2}{2}\right)\exp\left(\frac{-(z-\rho u)^2}{2(1-\rho^2)}\right)dz\,du\\
    &\qquad=\frac{1}{2\pi\sqrt{1-\rho^2}}\int_{\R^2}f\big(\Ubar^y(x,y,z)\big)
    \exp\left(\frac{-z^2}{2}\right)\exp\left(\frac{-(u-\rho z)^2}{2(1-\rho^2)}\right)dz\,du\\
    &\qquad=\frac{1}{\sqrt{2\pi}}\int_{\R}f\big(\Ubar^y(x,y,z)\big)\exp\left(\frac{-z^2}{2}\right) \underbrace{\frac{1}{\sqrt{2\pi(1-\rho^2)}}\int_{\R}\exp\left(\frac{-(u-\rho z)^2}{2(1-\rho^2)}\right)du}_{=1}\,dz\\
    &\qquad=\frac{1}{\sqrt{2\pi}}\int_{\R}f\big(\Ubar^y(x,y,z)\big)\exp\left(\frac{-z^2}{2}\right)dz,
\end{align*}
where we have used Fubini's theorem in the penultimate step. Thus, we obtain
\[
    \E{f\big(\U^y(x,y,Z^x_{k+1},Z^\perp_{k+1})\big)}=\E{f\big(\Ubar^y(x,y,Z)\big)}.
\]
Putting everything together, we have
\begin{align*}
    \widetilde{D}\big(\Gamma_{k+1}^y\big)&=\int_{\R^2}\E{f\big(\Ubar^y(x,y,Z)\big)}\,d\P(\Xd_k\leq x,\Yd_k\leq y)\\
    &=\int_{\R^2}\E{\big\|\Ubar^y(x,y,Z)-\Yq_{k+1}\big\|^2}\,d\P(\Xd_k\leq x,\Yd_k\leq y)\\
    &=\E{\big\|\overline{\U}^y(\Xd_k,\Yd_k,Z)-\Yq_{k+1}\big\|^2},
\end{align*}
as required.
\end{proof}

\begin{rem}
The above proposition demonstrates that the quantization of $\Yd_{k+1}$ depends only on its distribution, and, from the perspective of the distortion function, the correlation between $\Yd_{k+1}$ and $\Xd_{k+1}$ is irrelevant. Another way of saying this is that
\[
    f\big(\U^y(x,y,Z^x_{k+1},Z^\perp_{k+1})\big){\buildrel d \over =}f\big(\Ubar^y(x,y,Z)\big),
\]
where $Z\sim\mathscr{N}(0,1)$, and, since we only need to consider weighted sums of expectations of these values when computing the distortion, the correlation between $Z^x_{k+1}$ and $Z^\perp_{k+1}$ need not be considered. As we shall see later, it is necessary to take correlation into account when computing the joint probabilities of $\Yd_{k+1}$ and $\Xd_{k+1}$.
\end{rem}

As we did in the previous section, we now quantize the expression for the distortion. The quantization of the Euler scheme for the independent process, $\Xd$, proceeds directly using the standard RMQ algorithm from Section \ref{Sec: Single-factor}, and can be performed for all time steps without reference to $\Yd$. Suppose, at time step $k$, the quantizer for the dependent process $\Gamma^y_k$ has been computed along with the corresponding joint probabilities $\P(\Xq_k= x_k^i,\Yq_k= y_k^u)$, for $1\leq i \leq N^x$ and $ 1\leq u\leq N^y$, then the distortion for the quantizer of $\Yd_{k+1}$ may be approximated by
\begin{align}
    \widetilde{D}(\Gamma^y_{k+1}) &= \int_{\R^2}\E{\big\|\Ubar^y(x,y,Z)-\Yq_{k+1}\big\|^2}\,d\P(\Xd_k\leq x,\Yd_k\leq y)\notag\\
    &\approx\sum_{i=1}^{N^x} \sum_{u = 1}^{N^y}\E{(\overline{\U}^y(x^i_k, y^u_k, Z) - \Yq_{k+1})^2} \P(\Xq_k = x^i_k, \Yq_k = y^u_k) \label{Eq: SV Distortion}\\
    &=:D(\Gamma^y_{k+1})\notag.
\end{align}
The main result from \cite{pages2015recursive} shows that the approximation in \eqref{Eq: SV Distortion} results in a convergent procedure.
We again assume that the cardinality of $\Gamma_k^y$ is fixed at $N^y$ for all $0<k\leq K$ and that $\Gamma_0^y=\{y_0\}$. As before, the quantizer $\Gamma_1^y$ may be computed using standard vector quantization of the normal distribution.

For the remainder of this section we assume that, conditional on knowing the quantizers $\Gamma^x_{k}$ and $\Gamma^y_{k}$, their associated joint probabilities are known --- in Section \ref{Sec: Joint Probability} we shall provide two different approaches for computing them. Under this assumption and having rewritten the distortion \eqref{Eq: SV Distortion} in terms of the margined update function, the minimization problem that generates the quantizer at time-step $k+1$ may be specified using the Newton-Raphson iteration
\begin{equation}
    \prescript{(n+1)}{}{\mGamma}^y_{k+1} = \prescript{(n)}{}{\mGamma}^y_{k+1}-\left[\nabla^2 D\left(\prescript{(n)}{}{\mGamma}^y_{k+1}\right)\right]^{-1}\nabla D\left(\prescript{(n)}{}{\mGamma}^y_{k+1}\right), \label{Eq: Newton-Raphson}
\end{equation}
where $\mGamma^y_{k+1}$ is a column vector of the codewords in $\Gamma^y_{k+1}$ and $0\leq n<n_{\mathrm{max}}$ is the iteration index. Closely following \cite{mcwalter2017recursive}, closed-form expressions for the gradient of the distortion, $\nabla D\left(\mGamma^y_{k+1}\right)$, and the tridiagonal Hessian matrix, $\nabla^2 D\left(\mGamma^y_{k+1}\right)$, may now be derived.

To summarise notation, we write the update of the dependent process as
\[
    \overline{\U}^y(x^i_k,y^u_k, Z)=:U_{k+1}^{i,u}=\overline{m}^{i,u}_kZ + \overline{c}^u_k,
\]
where
\[
    \overline{m}^{i,u}_k:=b^y(x^{i}_{k}, y^u_k)\sqrt{\Delt}\qquad\text{and}\qquad
    \overline{c}^u_k:=y_{k}^{u}+a^y(y_k^u)\Delt.
\]
Note that the $i$ and $j$ indices, for $1\leq i,\,j\leq N^x$, always refer to the codewords of the quantizers for the $\Xd$-process, whereas the $u$ and $v$ indices, for $1\leq u,\,v\leq N^y$, always refer to the codewords of the quantizers for the $\Yd$-process.

The gradient of the distortion is given by
\begin{align}
    \frac{\partial D(\Gamma^y_{k+1})}{\partial y^v_{k+1}} &= 2 \sum_{i=1}^{N^x}\sum_{u=1}^{N^y} \E{\ind{U^{i,u}_{k+1}\in R^v_{k+1}}(y^v_{k+1} - U^{i,u}_{k+1})}\P(\Xq_k = x^i_k, \Yq_k = y^u_k) \notag \\
    &= 2 \sum_{i=1}^{N^x}\sum_{u=1}^{N^y} \int_{U^{i,u}_{k+1}\in R^v_{k+1}}(y^v_{k+1} - U^{i,u}_{k+1}) \,d\P(Z<z) \P(\Xq_k = x^i_k, \Yq_k = y^u_k)\label{gradint},
\end{align}
where $R^v_{k+1}$ is the region associated with codeword $y^v_{k+1}$.
To rewrite the integration bounds in terms of the Gaussian random variable, consider that $U_{k+1}^{i,u}\in R^v_{k+1}$ implies that $U_{k+1}^{i,u}$ lies between the region boundaries of the codeword $y^v_{k+1}$. This means
\[
    \Rm{k+1}{v}<U_{k+1}^{i,u}\leq\Rp{k+1}{v}\qquad\text{and}\qquad
    \Rpm{k+1}{v}:=\tfrac{1}{2}(y_{k+1}^{v\pm1}+y_{k+1}^{v}),
\]
and $\Rm{k+1}{1}=-\infty$ and $\Rp{k+1}{N^y}=\infty$ by definition.
Thus,
\[
    U_{k+1}^{i,u}\in R_{k+1}^v=
    \begin{cases}
        \Rm{k+1}{i,u,v}<Z\leq\Rp{k+1}{i,u,v} & \text{for $\overline{m}^{i,u}_k\geq 0$}\\[3mm]
        \Rm{k+1}{i,u,v}>Z\geq\Rp{k+1}{i,u,v} & \text{for $\overline{m}^{i,u}_k<0$,}
    \end{cases} 
\]
where
\begin{equation}
    r_{k+1}^{i,u,v\pm}:=\frac{\Rpm{k+1}{v}-\overline{c}^u_k}{\overline{m}^{i,u}_k}, \label{Eq: left boundary}
\end{equation}
is defined to be the standardized region boundary.
Similar to the region boundaries of the independent process, see \eqref{Eq: X Region Boundaries}, it refers to the region boundaries of the codeword $y^v_{k+1}$, when viewed from the codewords $x^i_{k}$ and $y^u_{k}$ of the previous time step.

Let $f_Z$ and $F_Z$ by the PDF and CDF of a standard normal random variable $Z$, respectively, and define $M_Z$ as the first lower partial expectation of $Z$,
\[
    M_Z(z):=\E{Z\ind{Z<z}}.
\]
Then, by direct evaluation of the integral in \eqref{gradint}, each element of the gradient of the distortion at time-step $k+1$ is given by
\begin{align}
\frac{\partial D(\Gamma^y_{k+1})}{\partial y^v_{k+1}} &=
2 \sum_{i=1}^{N^x}\sum_{u=1}^{N^y} \left[(y^v_{k+1} - \overline{c}^u_k)\sgn(\overline{m}^{i,u}_k)(F_{Z}(\Rp{k+1}{i,u,v}) - F_{Z}(\Rm{k+1}{i,u,v}))\right. \notag\\
&\qquad\qquad\left.-\abs{\overline{m}^{i,u}_k}(M_{Z}(\Rp{k+1}{i,u,v}) - M_{Z} (\Rm{k+1}{i,u,v}))\right] \P(\Xq_k = x^i_k, \Yq_k = y^u_k). \label{Eq: SV Gradient}
\intertext{The $N^y$-elements of the main diagonal of the tridiagonal Hessian matrix, $\nabla^2 D\left(\mGamma^y_{k+1}\right)$, are given by}
\frac{\partial^2 D(\Gamma^y_{k+1})}{\partial(y^v_{k+1})^2} &= \sum_{i=1}^{N^x}\sum_{u=1}^{N^y} \bigg[2\sgn(\overline{m}^{i,u}_k)(F_{Z}(\Rp{k+1}{i,u,v}) - F_{Z}(\Rm{k+1}{i,u,v})) \notag \\
&\qquad\qquad + \frac{1}{2\abs{\overline{m}^{i,u}_k}}f_{Z}(\Rp{k+1}{i,u,v})(y^v_{k+1} - y^{v+1}_{k+1})\\
&\qquad\qquad + \frac{1}{2\abs{\overline{m}^{i,u}_k}}f_{Z}(\Rm{k+1}{i,u,v})(y^{v-1}_{k+1} - y^{v}_{k+1})\bigg]\P(\Xq_k = x^i_k, \Yq_k = y^u_k),\notag
\intertext{with the $(N^y-1)$-elements of the super-diagonal and sub-diagonal given by}
\frac{\partial^2 D(\Gamma^y_{k+1})}{\partial y^v_{k+1} \partial y^{v+1}_{k+1}} &= \sum_{i=1}^{N^x}\sum_{u=1}^{N^y} \frac{1}{2\abs{\overline{m}^{i,u}_k}}f_{Z}(\Rp{k+1}{i,u,v})(y^v_{k+1} - y^{v+1}_{k+1}) \P(\Xq_k = x^i_k, \Yq_k = y^u_k)
\intertext{and}
\frac{\partial^2 D(\Gamma^y_{k+1})}{\partial y^v_{k+1} \partial y^{v-1}_{k+1}} &= \sum_{i=1}^{N^x}\sum_{u=1}^{N^y}  \frac{1}{2\abs{\overline{m}^{i,u}_k}}f_{Z}(\Rm{k+1}{i,u,v})(y^{v-1}_{k+1} - y^{v}_{k+1})\P(\Xq_k = x^i_k, \Yq_k = y^u_k), \label{Eq: SV Off Diag}
\end{align}
respectively.

The formulae above are similar to those derived for the standard RMQ case, with an additional summation over the codewords of the independent process. Thus, we again have a one-dimensional vector quantization problem, but this time the marginal distribution to be quantized consists of a sum of Euler updates that are weighted using joint probabilities. For this reason, we shall refer to this variant of the RMQ algorithm as the \emph{joint} RMQ algorithm (JRMQ). This will allow us to distinguish it in the text from the standard RMQ algorithm described in Section \ref{Sec: Single-factor}.

When the above formulation is compared with the approach proposed by \cite{callegaro2015pricing} (see Appendix D of their paper), it is observed that our equations have one less summation, since we do not need to condition on the independent process at time-step $k+1$. This means that the expressions for the gradient and Hessian presented here are an order of magnitude more efficient to implement.







%% file: Sections/JointProbability.tex
\input{Figures/Fig_joint_region.TpX}

Up to this point, we have assumed that the joint probabilities required in \eqref{Eq: SV Gradient} to \eqref{Eq: SV Off Diag} are available. In this section, we shall show how to compute these probabilities exactly and using a computationally efficient approximation. To facilitate efficient implementation, we also provide a matrix formulation of the system in Section \ref{Sec: SV Matrix Formulation}.

From \eqref{Eq: X-Euler} and \eqref{Eq: Y-Euler} it is evident that, conditional on the realizations of $\Xd_k$ and $\Yd_k$, the joint probability distribution of $\Xd_{k+1}$ and $\Yd_{k+1}$ is bivariate Gaussian. We define
\[Z^y_{k+1} := \rho Z^x_{k+1} + \sqrt{1 - \rho^2}  Z^\perp_{k+1}, \]
such that $\Yd_{k+1} = \Ubar(\Xd_k, \Yd_k, Z^y_{k+1})$.

Consider the joint probability of $\Xd_{k+1}$ and $\Yd_{k+1}$ in the form
\begin{align}
	&F_{\Xd_{k+1},\Yd_{k+1}}(x,y)\notag\\
    &\qquad = \int_{\R^2} \P(\U^x(r, Z^x_{k+1})\leq x, \Ubar^y(r, s, Z^y_{k+1})\leq y)\, d\P(\Xd_k\leq r, \Yd_k\leq s) \notag\\
	&\qquad \approx \sum_{i=1}^{N^x}\sum_{u=1}^{N^y} \P(\U^x(x^i_k, Z^x_{k+1})\leq x, \Ubar^y(x^i_k, y^u_k, Z^y_{k+1})\leq y)\P(\Xq_k = x^i_k, \Yq_k = y^u_k) \label{Eq: Marginal Joint Probability} \\
	&\qquad = \sum_{i=1}^{N^x}\sum_{u=1}^{N^y} \P\left(Z^x_{k+1} \leq \frac{x - c^i_k}{m^i_k}, Z^y_{k+1} \leq \frac{y - \cbar^u_k}{\mbar^{i,u}_k}\right) \P(\Xq_k = x^i_k, \Yq_k = y^u_k).\notag
\end{align}
The approximation in \eqref{Eq: Marginal Joint Probability} is formed by replacing the continuous, and unknown, distributions of $\Xd_k$ and $\Yd_k$ with the discrete, and known quantized distributions of $\Xq_k$ and $\Yq_k$, as in the standard RMQ case. The necessary joint probability is then given by
\begin{equation}
\begin{split}
&\P(\Xd_{k+1}=x^j_{k+1}, \Yd_{k+1}= y^v_{k+1})\\
&\qquad\qquad\qquad= \sum_{i=1}^{N^x}\sum_{u=1}^{N^y}\left[ \int_{\Rm{k}{i,u,v}}^{\Rp{k}{i,u,v}}\int_{\Rm{k}{i,j}}^{\Rp{k}{i,j}} \phi_2(x, y, \rho)\,dx\,dy\right]\P(\Xq_k = x^i_k, \Yq_k = y^u_k),
\end{split}\label{Eq: Exact Joint Probability Integral}
\end{equation}
where $\phi_2(x,y,\rho)$ is the bivariate Gaussian density function for two standard Gaussian random variables correlated by $\rho$.
The double integral above refers to the probability of a rectangle delimited by the standardized regions of $\Gamma^x_{k+1}$ and $\Gamma^y_{k+1}$, see Figure \ref{Fig: Plot_bivariate_regions}.
Therefore, for each $1\leq j \leq N^x$ and $1\leq v\leq N^y$,
\begin{equation}
\begin{split}
\P(\Xd_{k+1}=x^j_{k+1}, \Yd_{k+1}= y^v_{k+1}) &= \sum_{i=1}^{N^x}\sum_{u=1}^{N^y}\left[ \Phi_2\big(\Rp{k}{i,j}, \Rp{k}{i,u,v},\rho\big) - \Phi_2\big(\Rm{k}{i,j}, \Rp{k}{i,u,v},\rho\big)  \right.\\
&\qquad\qquad \left.- \Phi_2\big(\Rp{k}{i,j}, \Rm{k}{i,u,v},\rho\big) + \Phi_2\big(\Rm{k}{i,j}, \Rm{k}{i,u,v},\rho\big)\right]\\
&\qquad\qquad \times \P(\Xq_k = x^i_k, \Yq_k = y^u_k),
\end{split}
\label{Eq: Exact Joint Probability}
\end{equation}
where $\Phi_2(x, y,\rho)$ is the standard bivariate Gaussian cumulative distribution function with correlation $\rho$ evaluated at $x$ and $y$.

Given the quantizers at time $k$, the joint probability in \eqref{Eq: Exact Joint Probability} is exact. However, it requires the evaluation of the bivariate Gaussian distribution function. Although most programming languages have an efficient implementation of this function, it is significantly more expensive to compute than the univariate distribution. The joint probability can be approximated using only calls to the univariate Gaussian CDF by using quadrature to approximate the inner integral of \eqref{Eq: Exact Joint Probability Integral}.

While other approaches are possible, a simple quadrature rule is used by replacing $\Xd_{k+1}$ with its quantized version, $\Xq_{k+1}$, which is constant over the interval. Then \eqref{Eq: Exact Joint Probability Integral} becomes
\begin{align}
&\P(\Xd_{k+1}=x^j_{k+1}, \Yd_{k+1}= y^v_{k+1})\notag\\
&\qquad\approx \sum_{i=1}^{N^x}\sum_{u=1}^{N^y} \left[ \int_{\Rm{k}{i,u,v}}^{\Rp{k}{i,u,v}} \phi_2\bigl(y,\rho\big|\tfrac{x^j_{k+1} - c^i_k}{m^i_k}\bigr)dy \right]\P(\Xq_{k+1} = x^j_{k+1})\P(\Xq_k = x^i_k, \Yq_k = y^u_k)\notag\\
\begin{split}
&\qquad= \sum_{i=1}^{N^x}\sum_{u=1}^{N^y} \left[ F_Z\left(\frac{\Rp{k}{i,u,v}- \rho\tfrac{x^j_{k+1} - c^i_k}{m^i_k}}{\sqrt{1 - \rho^2}}\right)- F_Z\left(\frac{\Rm{k}{i,u,v}-\rho\tfrac{x^j_{k+1} - c^i_k}{m^i_k}}{\sqrt{1 - \rho^2}}\right)\right]\\
&\qquad\qquad\qquad\times\P(\Xq_{k+1} = x^j_{k+1})\P(\Xq_k = x^i_k, \Yq_k = y^u_k),
\end{split}\label{Eq: Joint Probability Approximation}
\end{align}
where $\phi_2(y,\rho|x)$ is the conditional bivariate Gaussian density. It is worthwhile to note that this approximation to the joint probability, although derived differently, is identical to that of \cite{callegaro2015pricing}.

The computational efficiency of this approximation is demonstrated in the Sections \ref{Sec: Pricing European Options} and \ref{Sec: Pricing Exotic Options}.

%% file: Sections/Implementation.tex
In this section a concise matrix formulation for the JRMQ algorithm is presented, similar to that provided in \cite{mcwalter2017recursive} for the standard RMQ case.

\subsection{Matrix Formulation}
\label{Sec: SV Matrix Formulation}
Throughout this section, the index $1\leq i\leq N^x$ refers to time-step $k$ and $1\leq j\leq N^x$ refers to time-step $k+1$, and both are associated with the $\Xd$-process. For the $\Yd$-process, the index $1\leq u\leq N^y$ refers to time-step $k$ and the index $1\leq v\leq N^y$ refers to time-step $k+1$.

To initialize the JRMQ algorithm, the standard RMQ algorithm is applied to the $\Xd$-process and yields the quantizers $\mGamma^x_k$ and associated probabilities $\mp^x_k$ at each time-step $0\leq k\leq K$. The following three variables are initialized
\[
    [\mGamma^y_0]_1 = y_0,\qquad[\mp_0^x]_1 = 1,\qquad[\mJnt_0]_{1,1} = 1,
\]
being the time-zero quantizer, associated probability and margined probability, respectively, of the $\Yd$-process. The standard one-dimensional vector quantization algorithm (on the normal distribution) is used to produce $\mGamma^y_1$ and $\mp^y_1$, being the quantizer and associated probability vector of the $\Yd$-process at the first time step. The corresponding joint probabilities at time-step one may then be computed using either \eqref{Eq: Joint Probability Matrix} or \eqref{Eq: New Joint Probability Matrix} with $k=0$ and $N^x=N^y=1$.

We now describe the implementation of the recursive step form time-step $k$ to $k+1$. Consider the time-step $k$ quantizers
\[
[\mGamma^x_k]_i =x^i_k\qquad\text{and}\qquad[\mGamma^y_k]_u = y^u_k,
\]
of the independent and dependent processes, respectively, and the associated joint probability matrix $\mJnt_k$, of size $N^x\times N^y$,
\[
[\mJnt_k]_{i,u} = \P(\Xq_k= x^i_k, \Yq_k = y^u_k),
\]
all of which are assumed known (already computed). The rows of $\mJnt_k$ are denoted by $\mJnt_k^{(i)}$.

The time-step $k+1$ quantizer for the dependent process and associated probabilities are computed as follows: Aside from an initial guess for $\mGamma^y_{k+1}$, which is taken to be $\mGamma^y_k$,
we initialize the $N^y$-element column vector
\begin{align*}
[\mc_k]_u &= \overline{c}^u_k
\intertext{and the set of $N^y$-element column vectors, indexed by $i$,}
[\mm_k]_u^{(i)} &= \overline{m}^{i,u}_k,
\end{align*}
in terms of the time-step $k$ quantities listed above. For each iteration of the Newton-Raphson algorithm, three sets of matrices, indexed by $i$, are computed. The first two sets have matrices of size $N^y \times N^y$, given by
\begin{align*}
[\mP_{k+1}]^{(i)}_{u,v} &= \P(\Yq_{k+1} = y^v_{k+1}|\Xq_k = x^i_k, \Yq_k = y^u_k), \notag \\
&= F_{Z}(\Rp{k+1}{i,u,v}) - F_{Z}(\Rm{k+1}{i,u,v})
\intertext{and}
[\mM_{k+1}]^{(i)}_{u,v} &= M_{Z}(\Rp{k+1}{i,u,v}) - M_{Z}(\Rm{k+1}{i,u,v}),
\intertext{while the third has matrices of size $N^y \times (N^y-1)$, given by}
[\mf_{k+1}]_{u,v}^{(i)} &= f_{Z}(\Rp{k+1}{i,u,v}).
\end{align*}
The above matrices allow the gradient and the Hessian of the distortion for $\mGamma^y_{k+1}$ to be written in simplified form. The $N^y$-element gradient vector is
\begin{align}
\nabla D(\mGamma^y_{k+1})^\top &= \sum_{i=1}^{N^x} 2\mJnt_k^{(i)}(((\mGamma^y_{k+1}\mJ_{N^y})^\top - \mc_k\mJ_{N^y}) \hprod \mP_{k+1}^{(i)} - (\abs{\mm_k^{(i)}}\mJ_{N^y})\hprod \mM_{k+1}^{(i)} ), \label{Eq: SV Matrix Gradient}
\intertext{where $\hprod$ is the Hadamard (or element-wise) product and $\mJ_z$ is defined to be a length-$z$ row vector of ones. By specifying the column vector}
[\Delta \mGamma^y_{k+1}]_v &= y^{v+1}_{k+1} - y^v_{k+1},
\intertext{with $1 \leq v \leq (N^y-1)$, the $(N^y-1)$-element off-diagonal of the tridiagonal Hessian matrix is given by }
\mho &= \sum_{i=1}^{N^x} -\frac{1}{2}\mJnt_k^{(i)}((\abs{\mm_k^{(i)}}^{\hprod - 1}\mJ_{N^y-1})\hprod \mf_{k+1}^{(i)} \hprod (\Delta \mGamma^y_{k+1} \mJ_{N^y})^\top)
\intertext{and the $N^y$-element main diagonal by}
\mhm &= \sum_{i=1}^{N^x} 2\mJnt_k^{(i)} \mP_{k+1}^{(i)} + [\mho | 0] + [0 | \mho]. \label{Eq: SV Matrix Main Diag}
\end{align}
Here, $\hprod - 1$ refers to the element-wise inverse.

Equations \eqref{Eq: SV Matrix Gradient} to \eqref{Eq: SV Matrix Main Diag} provide a matrix representation of equations \eqref{Eq: SV Gradient} to \eqref{Eq: SV Off Diag} and correspond to those in the matrix implementation of the single-factor RMQ case. This allows straightforward implementation of the Newton-Raphson algorithm described by \eqref{Eq: Newton-Raphson}, ultimately yielding $\mGamma^y_{k+1}.$ It remains to compute the necessary probabilities.

The elements of the joint probability matrix, $\mJnt_{k+1}$, at time-step $k+1$, are computed using the bivariate Gaussian distribution as
\begin{align}
[\mJnt_{k+1}]_{j,v} &=  \sum_{i=1}^{N^x}\sum_{u=1}^{N^y} \P(\Xq_{k+1} =x^j_{k+1}, \Yq_{k+1} = y^v_{k+1} | \Xq_{k} =x^i_{k}, \Yq_{k} = y^u_{k} ) \P(\Xq_k = x^i_k, \Yq_k = y^u_k) \notag \\
\begin{split}
&=\sum_{i=1}^{N^x}\sum_{u=1}^{N^y} \left( \Phi_2(\Rp{k}{i,j}, \Rp{k}{i,u,v},\rho) - \Phi_2(\Rm{k}{i,j}, \Rp{k}{i,u,v},\rho) \right.\\
&\qquad\qquad \left. - \Phi_2(\Rp{k}{i,j}, \Rm{k}{i,u,v},\rho) + \Phi_2(\Rm{k}{i,j}, \Rm{k}{i,u,v},\rho) \right) [\mJnt_k]_{i,u},
\end{split}\label{Eq: Joint Probability Matrix}
\intertext{with the probabilities associated with the new quantizer given by}
\mp^y_{k+1} &= \sum_{j=1}^{N^x} \mJnt_{k+1}^{(j)}.f \label{Eq: Quantizer Probabilities}
\intertext{Finally, to compute the transition probability matrix for the time-step $k+1$, it is necessary to recompute the $\mP_{k+1}$ matrix using the final regions associated with the new set of codewords at $k+1$. Then}
[\mP^y_{k+1}]_{u,v} &= \frac{\P(\Yq_k = y^u_k, \Yq_{k+1} = y^v_{k+1})}{\P(\Yq_k = y^u_k)} \notag \\
&= \frac{\sum_{i=1}^{N^x} \P(\Yq_{k+1} = y^v_{k+1} | \Xq_k = x^i_k, \Yq_k = y^u_k)\P(\Xq_k = x^i_k, \Yq_k = y^u_k)}{\P(\Yq_k = y^u_k)} \notag \\
&=\frac{\sum_{i=1}^{N^x} [\mP_{k+1}]^{(i)}_{u,v}[\mJnt_k]_{i,u}}{[\mp^y_k]_u}. \label{Eq: Transition Probability Matrix}
\end{align}

To compute the joint probabilities using the computationally efficient approximation instead of the bivariate Gaussian distribution, \eqref{Eq: Joint Probability Matrix} is replaced by
\begin{align}
[\mJnt_{k+1}]_{j,v} &= \sum_{i=1}^{N^x}\sum_{u=1}^{N^y} \left[ F_Z\left(\frac{\Rp{k}{i,u,v}- \rho\tfrac{x^j_{k+1} - c^i_k}{m^i_k}}{\sqrt{1 - \rho^2}}\right)- F_Z\left(\frac{\Rm{k}{i,u,v}-\rho\tfrac{x^j_{k+1} - c^i_k}{m^i_k}}{\sqrt{1 - \rho^2}}\right)\right] [\mp_{k+1}^x]_j [\mJnt_k]_{i,u}. \label{Eq: New Joint Probability Matrix}
\end{align}
The time-step $k+1$ quantizer probabilities and transition probability matrix, \eqref{Eq: Quantizer Probabilities} and \eqref{Eq: Transition Probability Matrix}, are now computed in terms of \eqref{Eq: New Joint Probability Matrix}.

\subsection{Zero Boundary Behaviour}
As in the standard RMQ algorithm, to correctly model the underlying processes it may be necessary to implement a reflecting or absorbing boundary at zero.
Both the dependent and independent process can be modified in this way, but, once the distribution of either process has been adjusted, it is difficult to compute the joint probabilities using the bivariate Gaussian distribution. Thus, the joint probability approximation \eqref{Eq: Joint Probability Approximation} is used.

\subsubsection*{The Independent Process}
In the stochastic volatility setting, the independent process represents the stochastic volatility or variance of the dependent process. This implies that it must remain strictly positive and thus it is only necessary to consider a reflecting boundary. In Monte Carlo simulation a reflecting boundary is modelled by the fully-truncated Euler scheme, shown to be the least-biased scheme for stochastic volatility models in \cite{lord2010comparison}\footnote{Note that a \emph{truncated} Euler scheme models \emph{reflecting} boundary behaviour.}.

Implementing a reflecting boundary in the standard RMQ case is discussed in detail in \cite{mcwalter2017recursive} and modifying the independent process in this way leaves the JRMQ algorithm unchanged.

\subsubsection*{The Dependent Process}
As the dependent process usually represents either an asset price or an interest rate, depending on the application, either a reflecting or absorbing boundary at zero may be appropriate. When the process modeled is an asset price, an absorbing boundary allows the possibility of bankruptcy, whereas a reflecting boundary is necessary to correctly model interest rates.

Modifying the algorithm to account for an absorbing boundary at zero is straightforward and incurs no additional computational burden. To ensure the non-negativity of the process, the domain of the marginal distribution implied by the quantizer at time $k$ is smaller than zero if
\[ Z < -\frac{\cbar_{k}^{u}}{\mbar_{k}^{i,u}}, \]
which implies that, under the requirement to ensure positive codewords at the next time step, the left-most region boundary must be set to
\[
    r_{k+1}^{i,u,1-}=-\frac{\cbar_{k}^{u}}{\mbar_{k}^{i,u}},
\]
for $1\leq i \leq N^x$ and $1\leq u \leq N^y$. This is equivalent to setting $\rm{1}_{k+1}=0$ and thus truncates the domain, cf.\ \eqref{Eq: left boundary}.

Truncating the domain of the implied marginal distribution at each time step will result in quantizers with probabilities that do not sum to one. This is because there is now effectively an additional codeword at zero, at which probability has accumulated. At the completion of the algorithm, the quantizers can be augmented with this additional codeword and its associated probability.

To model a reflecting boundary at zero, first the domain of the implied marginal distribution at each time step must be modified such that only positive codewords can be attained.
This is achieved by altering the left-most region boundary as above. Secondly, the distribution, density and lower partial expectation functions that appear in \eqref{Eq: SV Gradient} to \eqref{Eq: SV Off Diag} must be replaced by their reflected counterparts,
\begin{align*}
    {f}_{\overline{Z}^{i,u}_{k+1}}(y)&=f_{Z}(y)+f_{Z}(2\bar{y}^{i,u}_k-y),\notag\\
 	 {F}_{\overline{Z}^{i,u}_{k+1}}(y)&=F_{Z}(y)-F_{Z}(2\bar{y}^{i,u}_k-y),\notag\\
\intertext{and}
    {M}_{\overline{Z}^{i,u}_{k+1}}(y)&=M_{Z}(y)+M_{Z}(2\bar{y}^{i,u}_k-y)- 2\bar{y}^{i,u}_kF_{Z}(2\bar{y}^{i,u}_k-y), 
\end{align*}
for $y\in[\bar{y}^{i,u}_k, \infty)$, where	 $\bar{y}^{i,u}_k=-\tfrac{\cbar_{k}^{u}}{\mbar_{k}^{i,u}}$.
Note that these functions have an $i$ and $u$ subscript, indicating that they will be different for each term in the summations of \eqref{Eq: SV Gradient} to \eqref{Eq: SV Off Diag}.

%% file: Sections/Europeans.tex

\begin{figure}[t!]
     \begin{center}
         \includegraphics[width=\columnwidth]{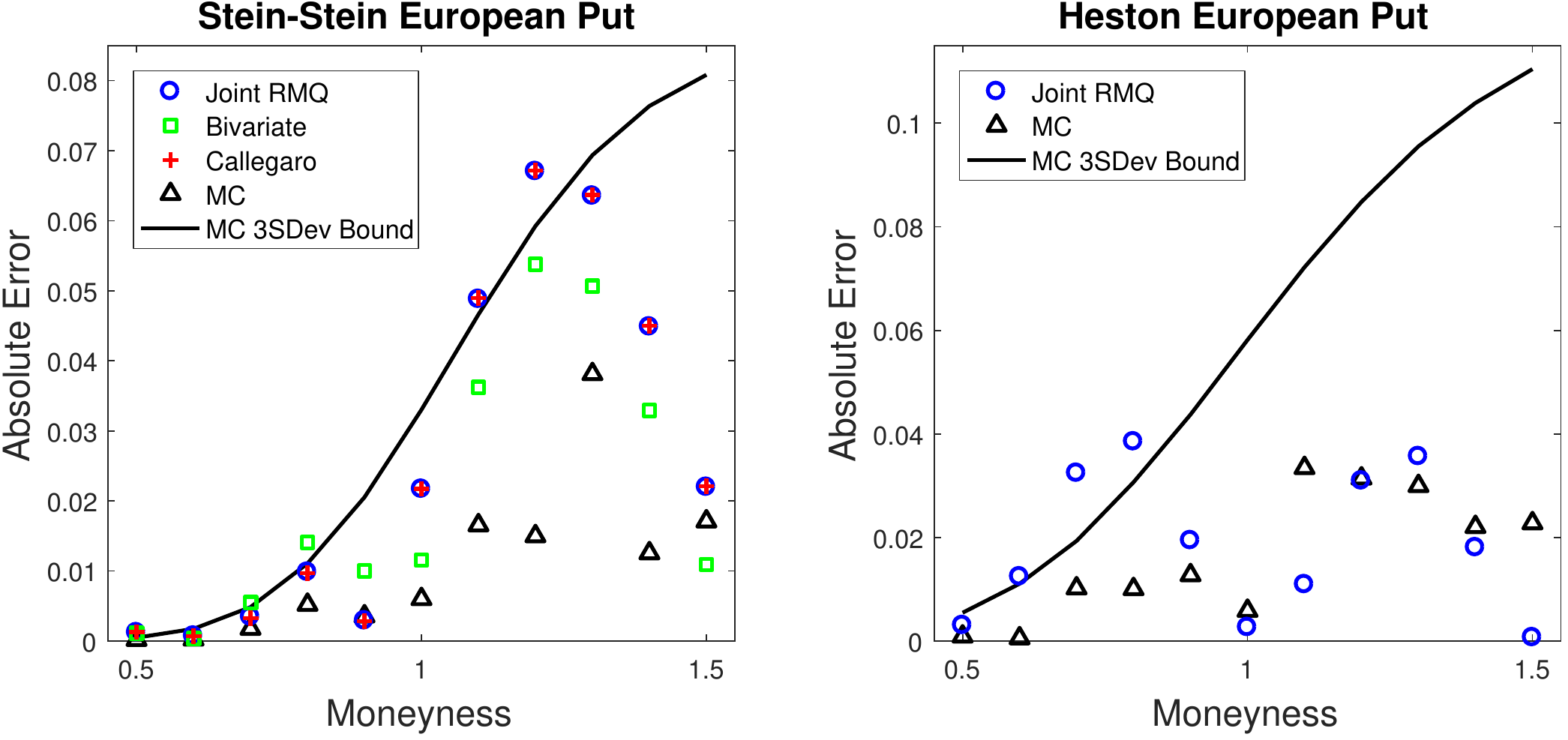}
     \end{center}
     \caption{The European put pricing error under the Stein-Stein and Heston models.}
     \label{Fig: Plot_Stein_Stein_Heston}
\end{figure}

In this section, we consider the pricing of European options under the \cite{stein1991stock}, \cite{heston1993closed} and SABR \citep{hagan2002managing} models. The Stein-Stein and Heston models are both amenable to semi-analytical pricing using Fourier transform techniques, whereas an analytical approximation exists for both the Black and Bachelier implied volatility under the SABR model. The Fourier pricing technique implemented uses the little trap formulation of the characteristic function from \cite{albrecher2006little} for the Heston model, while the \cite{schobel1999stochastic} characteristic function formulation is used for the Stein-Stein model. The implied volatility approximation for the SABR model is the latest from \cite{hagan2016universal}.

The Stein-Stein example is used to illustrate the computational efficiency advantage of the new algorithm compared to the RMQ algorithm from \cite{callegaro2015pricing}, whereas the Heston example serves to highlight the effectiveness of correctly modelling the zero-boundary behaviour of the independent process. For the SABR model, parameter sets were chosen that are difficult to handle with traditional methods, illustrating the flexibility of the JRMQ algorithm.

All simulations were executed using MATLAB 2016b on a computer with a $2.00$ GHz Intel i-$3$ processor and $4$ GB of RAM. All Monte Carlo simulations in this section used $500\,000$ paths with $120$ time steps per path.

\begin{figure}
     \begin{center}
         \includegraphics[width=\columnwidth]{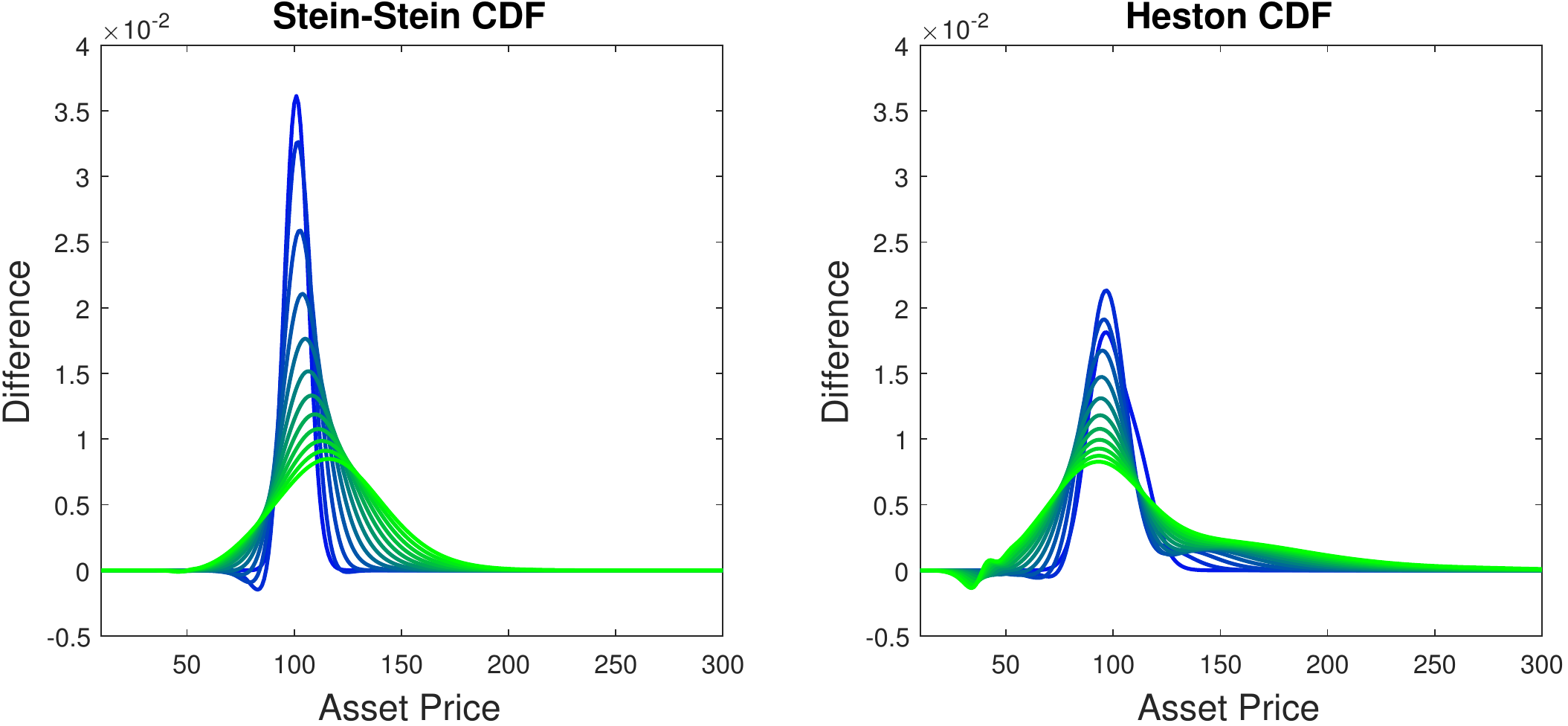}
     \end{center}
     \caption{The error in the marginal distributions for the dependent process in the Stein-Stein and Heston models.}
     \label{Fig: Plot_Stein_Stein_Heston_cdf}
\end{figure}

\subsection{The Stein and Stein Model}
\begin{figure}[t!]
     \begin{center}
         \includegraphics[width=\columnwidth]{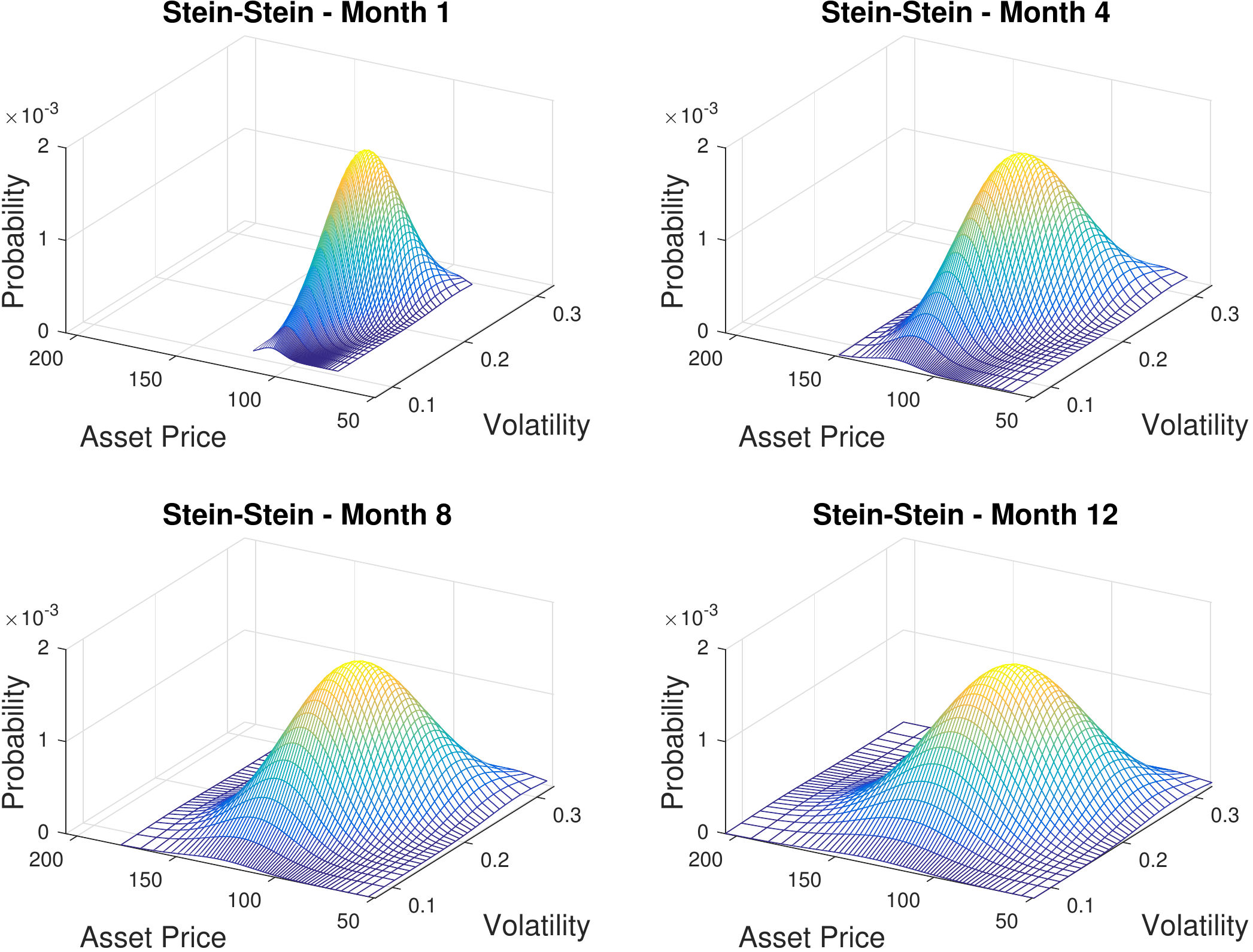}
     \end{center}
     \caption{Evolution of the approximate joint probability for the Stein-Stein model.}
     \label{Fig: Plot_Stein_Stein_joint}
\end{figure}

The SDEs for the Stein-Stein model may be specified in the notation of \eqref{Eq: X-process} and \eqref{Eq: Y-process} as
\begin{align*}
&a^x(X_t) = \kappa(\theta - X_t), & &b^x(X_t) = \sigma, \\
&a^y(Y_t) = rY_t, & &b^y(X_t, Y_t) = X_tY_t ,
\end{align*}
and in the example considered the parameters chosen are $\kappa = 4$, $\theta = 0.2$, $\sigma = 0.1$, $r = 0.0953$, $\rho=-0.5$, $x_0 = 0.2$ and $y_0=100$, with the maturity of the option set at one year. These parameters are from Table 1 in \cite{schobel1999stochastic}.

The left graph in Figure \ref{Fig: Plot_Stein_Stein_Heston} displays the pricing error of four algorithms. The first is the JRMQ algorithm presented in this paper using the joint probability approximation from \eqref{Eq: Joint Probability Approximation}, the second is the JRMQ algorithm using the bivariate Gaussian distribution, the third is the stochastic volatility RMQ algorithm from \cite{callegaro2015pricing} and the fourth is a two-dimensional standard Euler Monte Carlo simulation.

For the RMQ algorithms, $K = 12$ time steps were used with $N^x = 30$ codewords at each step for the independent process and $N^y=60$ codewords for the dependent process. We consider variable moneyness by changing the strike over the fixed initial asset price.

The JRMQ algorithm took $3.8$ seconds to price all strikes when using the probability approximation and $77.2$ seconds when using the bivariate Gaussian distribution. The algorithm from \cite{callegaro2015pricing} took $26.3$ seconds to price all strikes and the Monte Carlo simulation took $6.6$ seconds per strike.

The computation time of the JRMQ algorithm for this example was approximately 7 times faster than the algorithm of \cite{callegaro2015pricing}, when using approximate joint probabilities. Despite this large decrease in computation time, the JRMQ algorithm prices with the same accuracy. Barring three points, both algorithms price to within the three standard deviation bound of the significantly higher resolution Monte Carlo simulation. Using the bivariate Gaussian distribution instead of the approximation significantly reduces the average error over the range of moneyness considered, but this is at the expense of a large increase in computation time. For this reason, the remaining applications use only the approximation.

Since the Stein-Stein model has a closed-form characteristic function, it is possible to compute the marginal distribution for the dependent process. The difference between this marginal distribution and the one computed using the  JRMQ algorithm is presented in the left graph in Figure \ref{Fig: Plot_Stein_Stein_Heston_cdf}. The curve is blue at the initial time and changes color to green as we move toward maturity. The maximum error is under $4\%$ initially and decays to well under $1\%$ as time advances. These errors are in line with those of the one-dimensional Euler RMQ case illustrated in \cite{mcwalter2017recursive}.

Figure \ref{Fig: Plot_Stein_Stein_joint} illustrates the evolution of the approximate joint probabilities over time. Note that these are the joint probabilities associated with the quantizers and thus the grid is not uniform; there are $60$ points along the asset price axis and $30$ points along the volatility axis.

\subsection{The Heston Model}
\begin{figure}[t!]
     \begin{center}
         \includegraphics[width=\columnwidth]{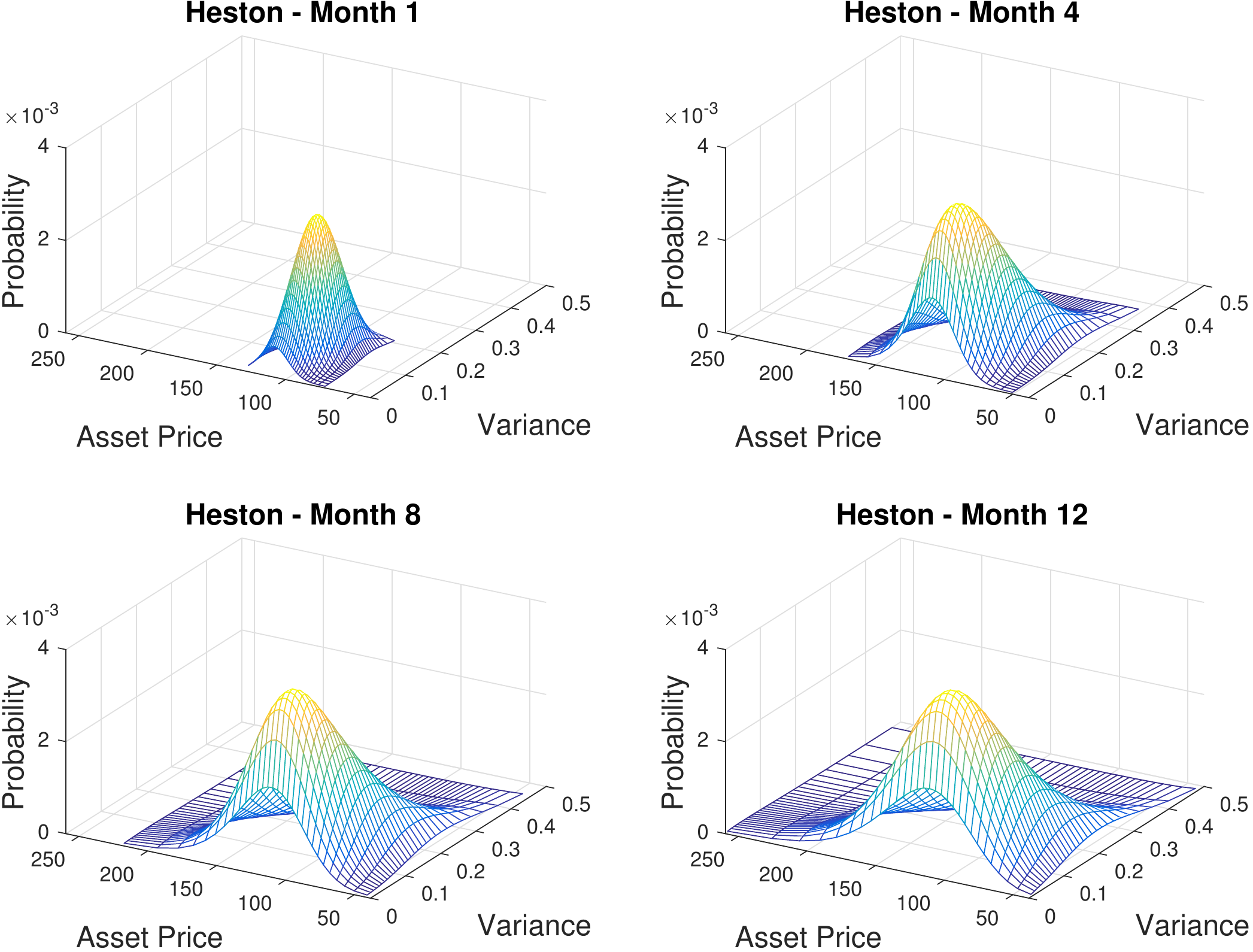}
     \end{center}
     \caption{Evolution of the approximate joint probability for the Heston model.}
     \label{Fig: Plot_Heston_joint}
\end{figure}

The SDEs for the Heston model may be specified in the notation of \eqref{Eq: X-process} and \eqref{Eq: Y-process} as
\begin{align*}
&a^x(X_t) = \kappa(\theta - X_t), & &b^x(X_t) = \sigma\sqrt{X_t}, \\
&a^y(Y_t) = rY_t, & &b^y(X_t, Y_t) = \sqrt{X_t}Y_t,
\end{align*}
and in the example considered the parameters chosen are $\kappa = 2$, $\theta = 0.09$, $\sigma = 0.4$, $r = 0.05$, $\rho=-0.3$, $x_0 = 0.09$ and $y_0=100$, with the maturity of the option set at one year. These parameters are based on the SV-I parameter set from Table 3 of \cite{lord2010comparison}, with $\sigma$ adjusted from $1$ to $0.4$ to ensure that the Feller condition is satisfied for the square-root variance process.

The right graph in Figure \ref{Fig: Plot_Stein_Stein_Heston} displays the pricing error for JRMQ compared with a two-dimensional fully truncated log-Euler scheme, suggested as the least-biased Monte Carlo scheme for stochastic volatility models in \cite{lord2010comparison}. For the JRMQ algorithm, $K = 12$ time steps were used with $N^x = N^y = 30$ codewords at each step for both processes. The JRMQ algorithm took $1.4$ seconds to price all strikes, whereas the Monte Carlo simulation took $7.8$ seconds for a single strike.

A reflecting zero-boundary was used when computing the standard RMQ algorithm for the independent variance process.
Compared to a high-resolution Monte Carlo simulation, the JRMQ algorithm performs very well despite the coarseness of the grid.

Even though the Feller condition is satisfied, due to the discretization of time, there is a non-zero probability of the Euler approximation for the variance process becoming negative. This is handled in the RMQ algorithm by using a reflecting zero-boundary. Modelling the boundary in this way leads to an increased accuracy in pricing, especially when compared to the Monte Carlo simulation.

The right graph in Figure \ref{Fig: Plot_Stein_Stein_Heston_cdf} presents the error in the marginal distribution of the dependent process implied by the RMQ algorithm when compared to the distribution obtained from the characteristic function using the Fourier transform technique. The error is just over $2\%$ initially and decreases to below $1\%$ as time advances. Figure \ref{Fig: Plot_Heston_joint} illustrates the evolution of the joint probabilities of the asset price and variance process.

\subsection{The SABR Model}
\begin{figure}[t!]
     \begin{center}
         \includegraphics[width=\columnwidth]{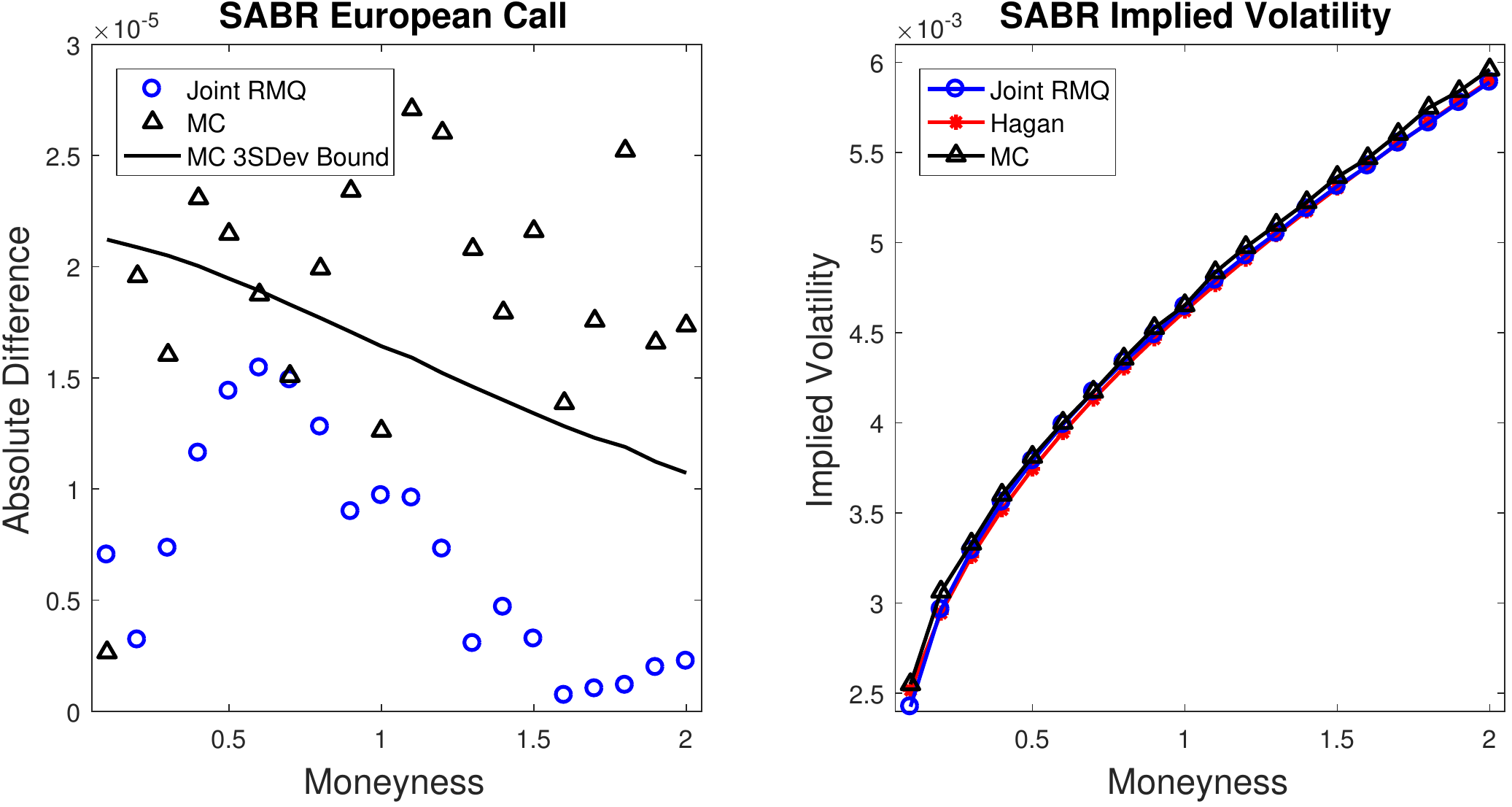}
     \end{center}
     \caption{Prices and implied Bachelier volatilities for the standard SABR model, using a parameter set applicable for interest rates.}
     \label{Fig: Plot_chen_SABR}
\end{figure}

The SDEs for the standard SABR model may be specified in the notation of \eqref{Eq: X-process} and \eqref{Eq: Y-process} as
\begin{align*}
&a^x(X_t) = 0, & &b^x(X_t) = \nu X_t, \\
&a^y(Y_t) = 0, & &b^y(X_t, Y_t) = X_t Y_t^\beta,
\end{align*}
with $0\leq\beta\leq1$.
A partial reason for the popularity of the SABR model is that the implied volatility may be computed using an analytical approximation \citep{hagan2002managing}. Further work has extended the original formula (see, for example, \cite{obloj2007fine} and \cite{paulot2015asymptotic}), with
the latest and most accurate approximation given in \cite{hagan2016universal}, which allows a more general specification of the volatility function.

In this section, we consider European options for two examples of extreme parameter sets that may arise in the context of interest rate modelling.

In Figure \ref{Fig: Plot_chen_SABR} the parameters chosen are $\beta = 0.7$, $\nu=0.3$, $\rho=-0.3$, $x_0 = 20\%$ and $y_0=0.5\%$, with the maturity of the option set at one year.
This parameter set is Test Case III from \cite{chen2012efficient}, and was specifically chosen to be appropriate to the fixed income market and to illustrate the correct handling of zero-boundary behaviour. The reference price is the implied volatility formula with the boundary correction from \cite{hagan2016universal}.

For the JRMQ algorithm, $K=24$ time steps were used with $N^x = N^y = 30$ codewords at each step for both processes. A reflecting zero-boundary was implemented for the dependent process. The Monte Carlo simulation utilized a fully-truncated Euler discretization scheme.

The three standard deviation bound in the left graph in Figure \ref{Fig: Plot_chen_SABR} indicates that the  Monte Carlo simulation is not converging to the same result as the \cite{hagan2016universal} implied volatility, used here as the reference price.
In their discussion, \cite{chen2012efficient} indicate that this is a challenging parameter set for traditional Monte Carlo simulations. Barring a single point, the JRMQ algorithm is more accurate than the Monte Carlo simulation across the range of strikes. It is also significantly faster to compute. The JRMQ algorithm took $5.3$ seconds to price all strikes, whereas the Monte Carlo simulation took $13.4$ seconds per strike, due to the much larger number of time steps.


\begin{figure}[t!]
     \begin{center}
         \includegraphics[width=\columnwidth]{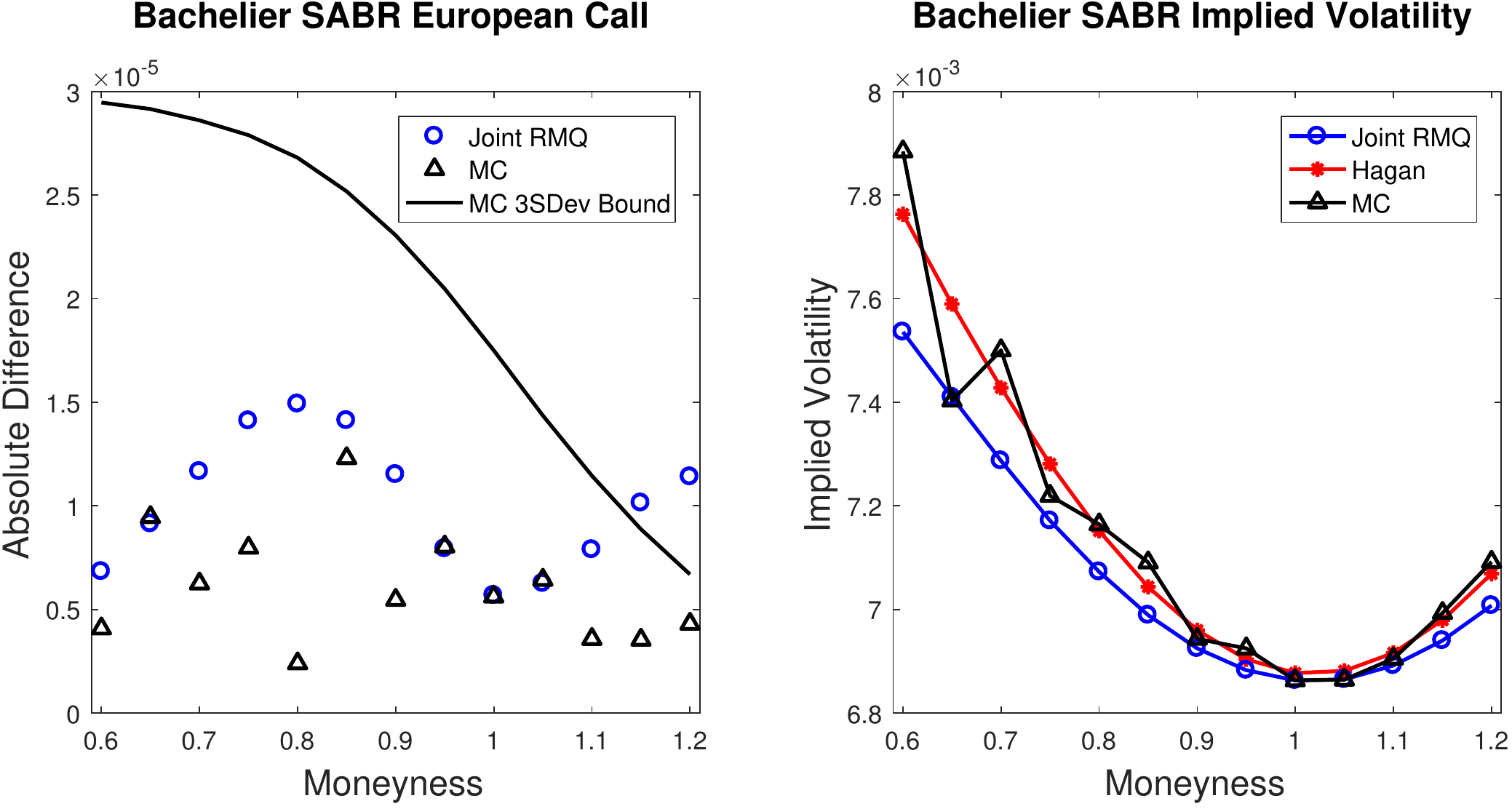}
     \end{center}
     \caption{Implied Bachelier volatility and pricing error for the Bachelier SABR model.}
     \label{Fig: Plot_normal_SABR}
\end{figure}

In Figure \ref{Fig: Plot_normal_SABR}, European call option prices and corresponding implied Bachelier volatilities are displayed for
the RMQ algorithm, the Hagan implied volatility approximation, and an Euler Monte Carlo simulation. The parameters chosen are $\beta = 0$, $\nu=0.3691$, $\rho=-0.0286$, $X_0 = 0.68\%$, $Y_0=4.35\%$, with the maturity of the option set at one year. This parameter set is Test Case I from \cite{korn2013exact} and it describes a challenging simulation environment with a low initial forward rate which is very volatile.

For the JRMQ algorithm, $K = 24$ time steps were used with $N^x = 10$ codewords at each step for the independent process and $N^y=90$ codewords for the dependent process. The JRMQ algorithm took $5.5$ seconds to price all strikes, whereas the Monte Carlo simulation took $5.6$ seconds per strike.

Despite the extreme parameter set, all but two of the JRMQ prices fall well within the three standard deviation bound of the much higher resolution Monte Carlo simulation.

%% file: Sections/Exotics.tex
\begin{figure}[t!]
     \begin{center}
         \includegraphics[width=\columnwidth]{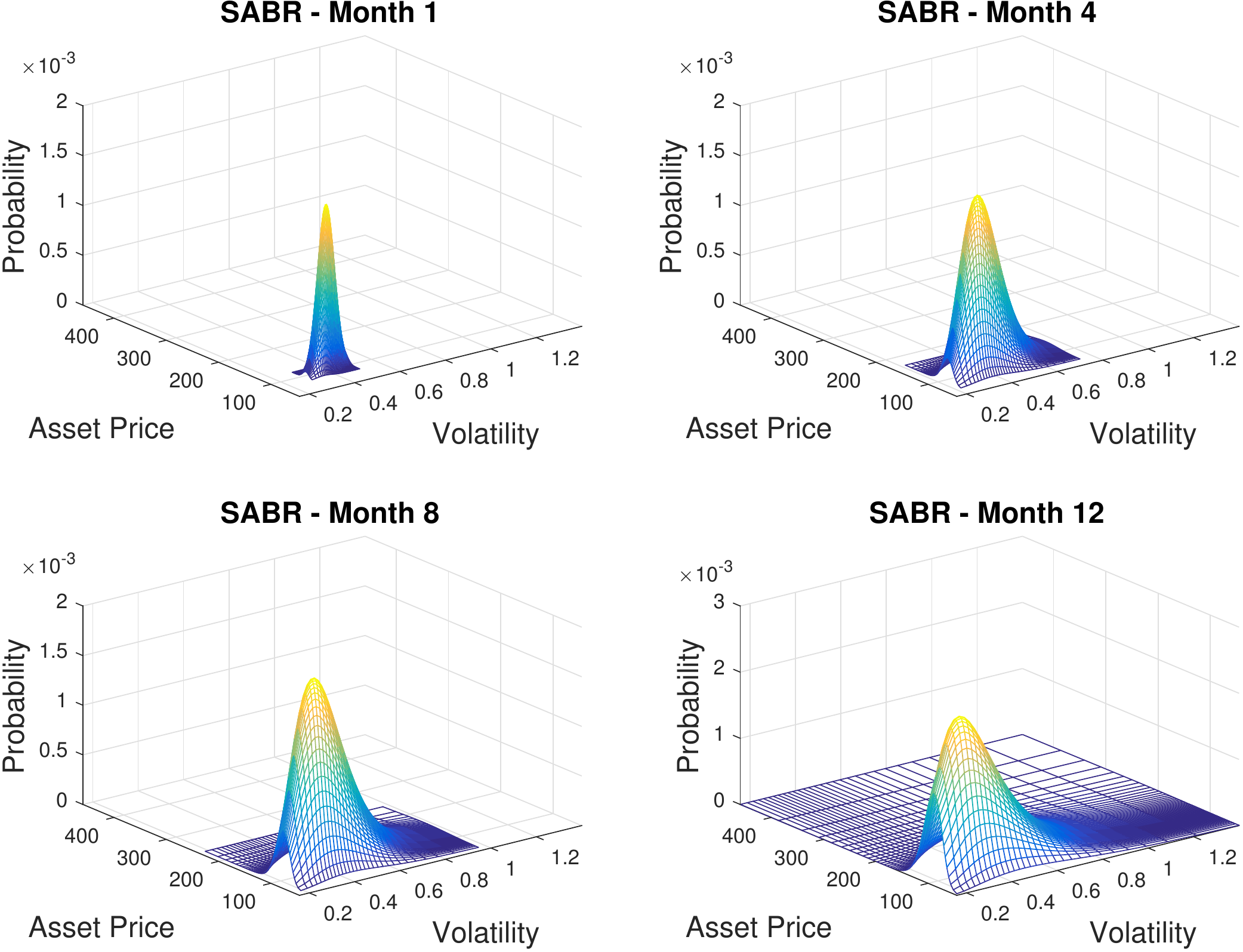}
     \end{center}
     \caption{Time-evolution of the approximate joint probability for the SABR model.}
     \label{Fig: Plot_joint_prob_SABR}
\end{figure}

An advantage of the RMQ algorithm, similar to binomial and trinomial tree methods, is the ability to price many options off the same grid that results from a single run. This is demonstrated in this section by using a single pass of the JRMQ algorithm to price European, Bermudan and barrier options, and volatility corridor swaps.

The SABR model parameters for all the examples in this section are $\beta = 0.9$, $\nu=0.4$, $\rho=-0.3$, $X_0 = 0.4$ and $Y_0=S_0\exp(rT)$, where $Y$ now models the $T$-forward price of an equity asset with $S_0 = 100$, $r=0.05$ and the maturity $T$ is equal to one year. The JRMQ algorithm used $K = 24$ time steps with $N^x = 30$ codewords for the volatility process and $N^y = 60$ codewords for the forward price process.
The Monte Carlo simulations are executed using a fully-truncated Euler scheme with $500\,000$ paths and $120$ time-steps.

To generate the quantization grid, the JRMQ algorithm took $7.8$ seconds for these parameters. The computational cost of generating derivative prices using the resulting grid is negligible in comparison. Figure \ref{Fig: Plot_joint_prob_SABR} illustrates the time-evolution of the approximate joint probability of the forward price of the asset and the volatility over the course of the year.

\begin{figure}[t!]
     \begin{center}
         \includegraphics[width=\columnwidth]{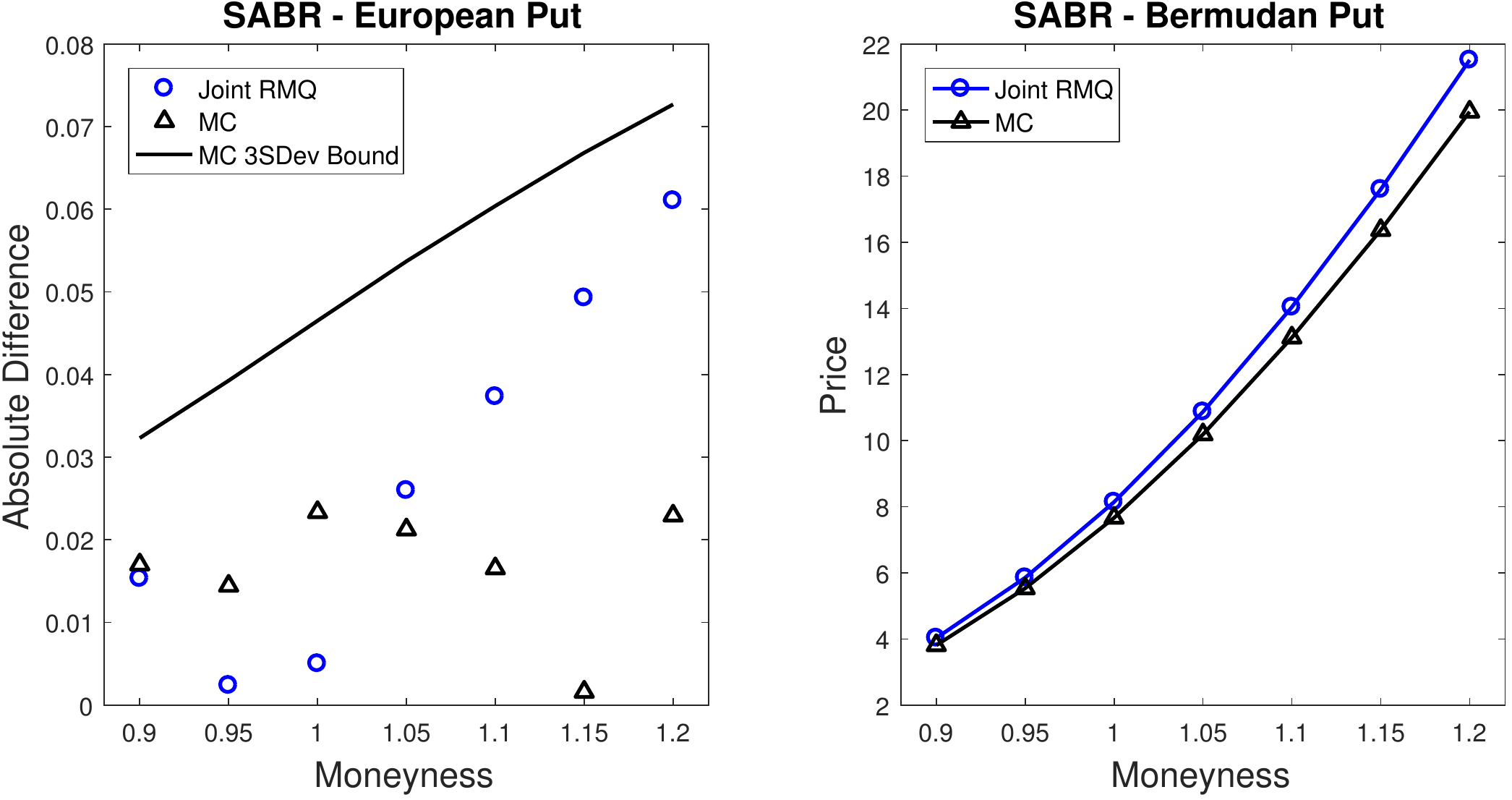}
     \end{center}
     \caption{European and Bermudan put option prices for the SABR model.}
     \label{Fig: Plot_tian_SABR1}
\end{figure}

The left graph in Figure \ref{Fig: Plot_tian_SABR1} illustrates the difference in the prices of European put options using JRMQ and the prices using the implied volatility formula of \cite{hagan2016universal}. The right graph shows the prices for a Bermudan put with monthly exercise opportunities using JRMQ and a least-squares Monte Carlo simulation. For each strike, computing an option price using Monte Carlo simulation takes approximately $14.5$ seconds for the European options and $16.9$ seconds for the Bermudan options. The high-level algorithm for pricing Bermudan options using a quantization grid is outlined in \cite{mcwalter2017recursive}.

\begin{figure}[t!]
     \begin{center}
         \includegraphics[width=\columnwidth]{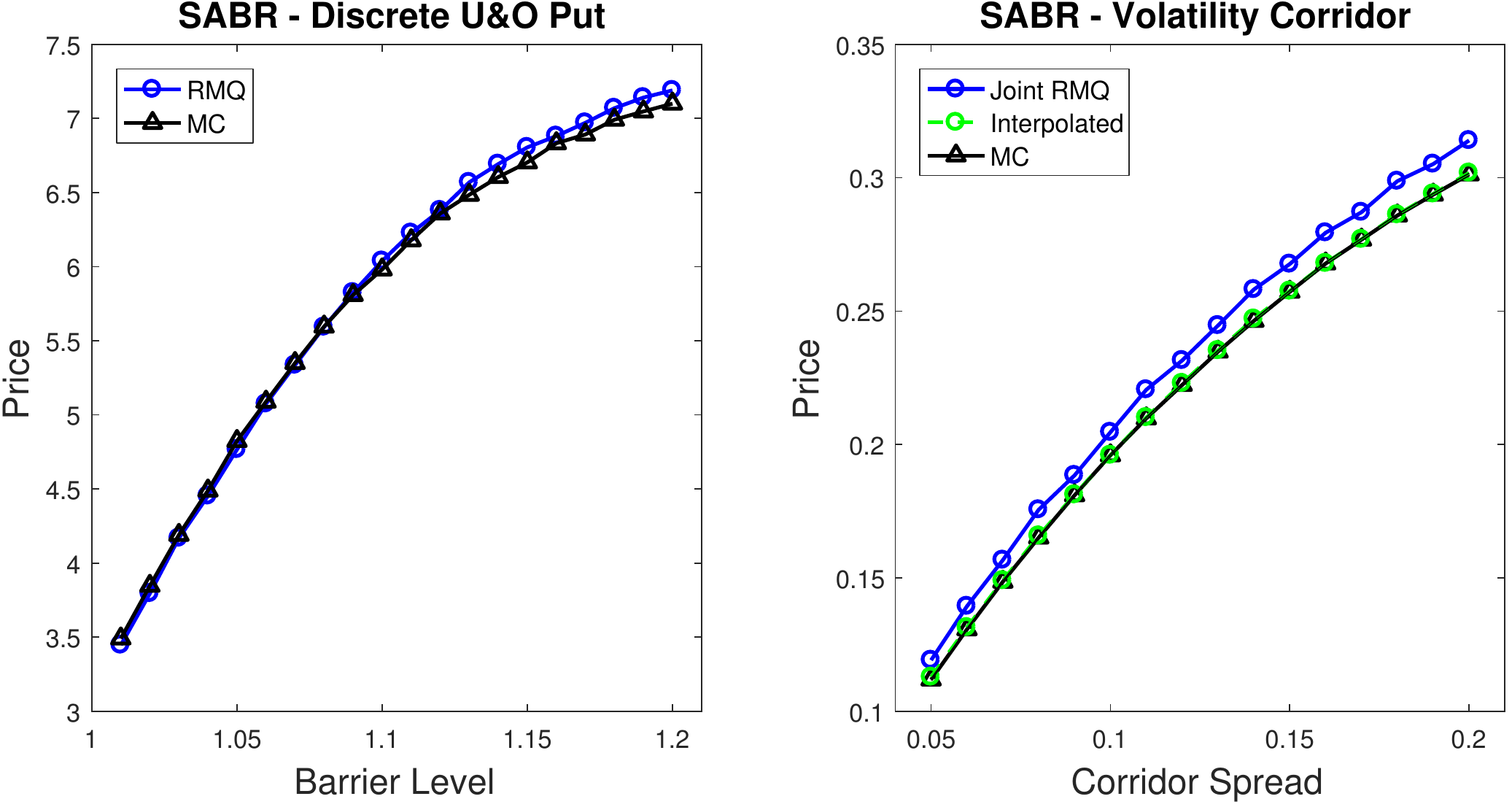}
     \end{center}
     \caption{Price comparison for discrete up-and-out put options and volatility corridor swaps in the SABR model.}
     \label{Fig: Plot_tian_SABR2}
\end{figure}

The left graph in Figure \ref{Fig: Plot_tian_SABR2} shows the JRMQ and Monte Carlo prices for a discrete up-and-out put option, with monthly monitoring, where the barrier level is expressed as a multiple of the at-the-money strike. The right graph shows the prices for a series of volatility corridor swaps. The payoff of a volatility corridor swap is given by
\begin{align}
 \frac{1}{T}\int_0^T X_z \ind{L < S_z < H}\, dz, \label{Eq: Volatility Corridor Swap Payoff}
\end{align}
where $S_t = Y_t\exp(-r(T - t))$ is the asset price in our deterministic interest-rate framework and $L$ and $H$ specify the corridor of the asset price in which the volatility is accumulated. The algorithm for pricing volatility corridor swaps on a quantization grid in the stochastic volatility setting is presented in \cite{callegaro2015pricing} and uses a left-endpoint approximation to the integral in \eqref{Eq: Volatility Corridor Swap Payoff}

The corridor spreads on the $x$-axis represent a percentage bound around the initial asset price value, i.e., the lower bound of the corridor is given by $L = S_0(1 - s)$ and the upper bound by $H = S_0(1 + s)$, where $s$ is the corridor spread. The vertical gap between the prices generated by the Monte Carlo simulation and the RMQ algorithm is partially due to the increased accuracy of the Monte Carlo simulation when using simple quadrature to approximate \eqref{Eq: Volatility Corridor Swap Payoff}, as a result of the large number of time steps used. For a single barrier value or a single corridor spread, the Monte Carlo simulation takes approximately $15.2$ seconds and $16.3$ seconds to price these derivatives.

The accuracy of JRMQ volatility corridor swap prices can be improved without using additional time steps. An increase in the accuracy of the approximation to the integral \eqref{Eq: Volatility Corridor Swap Payoff} is achieved by interpolating both the asset price and the volatility over each interval, see Appendix \ref{Sec: Volatility Corridor Swaps}. The improved accuracy of this interpolated JRMQ price is displayed in the right graph of Figure \ref{Fig: Plot_tian_SABR2}. 

%% file: Sections/Conclusion.tex

In this work, we present a Joint Recursive Marginal Quantization algorithm for stochastic volatility models that provides a significant computational advantage over the most recent developments in this area.

The central idea is to margin over, and effectively \emph{undo}, the Cholesky decomposition in the two-dimensional Euler scheme when performing the quantization. We show how the joint probabilities can be computed exactly and using a computationally efficient approximation.

A concise matrix formulation was provided for efficient implementation. The robustness of the algorithm was demonstrated by pricing options with path dependencies, early exercise boundaries and exotic features. Parameter sets that would be appropriate to interest rate and equity environments were used to demonstrate the correct handling of the boundary behaviour.

JRMQ was shown to be accurate and fast when compared to traditional Monte Carlo methods. This will allow the calibration of large derivative books, as per \cite{Callegaroetal2015a}, to be extended from only considering local volatility models to the more flexible stochastic volatility models, while retaining the efficiency of the original recursive marginal quantization algorithm. 

%% file: Sections/Appendix.tex

\titleformat{\section}{\normalfont\Large\bfseries}{Appendix~\thesection: }{0em}{}

\section{Volatility Corridor Swaps}
\label{Sec: Volatility Corridor Swaps}
\input{Figures/Fig_interpolated.TpX}
Consider
\begin{align}
\frac{1}{T}\int_0^T X_z \ind{L < S_z < H}\, dz  &= \frac{1}{T}\sum_{k=0}^{K-1}\int_{t_k}^{t_{k+1}} X_z \ind{L < S_z < H}\, dz  \notag \\
&\approx \frac{1}{T}\sum_{k=0}^{K-1}\int_{t^*(S_{t_k})}^{t^*(S_{t_{k+1}})} \frac{X_{t_{k+1}} - X_{t_k}}{\Delt}(z - t_k) + X_{t_k}\, dz, \label{Eq: Integral Approx}
\end{align}
where the volatility process $X_t$ has been replaced by a linear interpolation on the interval $t\in[t_k, t_{k+1}]$ with
\begin{align*}
t^*(s) &= \begin{cases}
t_k& \text{if $L\leq s \leq H$ and $s = S_{t_k}$,}\\
t_{k+1}& \text{if $L\leq s \leq H$ and $s = S_{t_{k+1}}$,}\\
\frac{H - S_{t_k}}{S_{t_{k+1}} - S_{t_k}}\Delt + t_k& \text{if $s > H$,}\\
\frac{L - S_{t_k}}{S_{t_{k+1}} - S_{t_k}}\Delt + t_k& \text{if $s < L$,}
\end{cases}
\end{align*}
providing the intercepts of the line connecting $S_{t_k}$ and $S_{t_{k+1}}$ with the corridor. This interpolation is illustrated in Figure \ref{Fig: Interpolation} and accounts for the indicator function by constraining the integration to where the asset price is in the corridor. Explicitly computing a single term from \eqref{Eq: Integral Approx} gives
\begin{multline*}
G(t_k, t_{k+1}, X_{t_k}, S_{t_k}, X_{t_{k+1}}, S_{t_{k+1}}) := \\
\frac{X_{t_{k+1}} - X_{t_k}}{2\Delt}\left[(t^*(S_{t_{k+1}}) - t_k)^2 - (t^*(S_{t_k}) - t_k)^2\right]
+ X_{t_k}\left[t^*(S_{t_{k+1}}) - t^*(S_{t_k})\right].
\end{multline*}
The value of a volatility corridor swap can now be computed as the expectation under the risk-neutral measure approximated using the quantization grids for $X$ and $S$,
\begin{align*}
\E{\frac{1}{T}\int_0^T X_z \ind{L < S_z < H}\, dz} \approx & \frac{1}{T}\sum_{k=0}^{K-1}\sum_{i=1}^{N^x}\sum_{j=1}^{N^x}\sum_{u=1}^{N^y}\sum_{v=1}^{N^y}G(t_k, t_{k+1}, x^i_k, s^u_k,  x^j_{k+1}, s^v_{k+1}) \notag \\
&\times \P(\Xq_k = x^i_k, \widehat{S}_k = s^u_k, \Xq_{k+1} = x^j_{k+1}, \widehat{S}_{k+1} = s^v_{k+1}),
\end{align*}
with the probability
\begin{multline*}
\P(\Xq_k = x^i_k, \widehat{S}_k = s^u_k, \Xq_{k+1} = x^j_{k+1}, \widehat{S}_{k+1} = s^v_{k+1}) = \\
\P(\widehat{S}_{k+1} = s^v_{k+1}| \Xq_k = x^i_k, \widehat{S}_k = s^u_k, \Xq_{k+1} =x^j_{k+1}) \P(\Xq_k = x^i_k, \widehat{S}_k = s^u_k, \Xq_{k+1} =x^j_{k+1}),
\end{multline*}
computed as part of the matrix formulation in Section \ref{Sec: Implementation}.

%% file: Fast_Quantization_of_SV_Models.bbl
\begin{thebibliography}{22}
\providecommand{\natexlab}[1]{#1}
\providecommand{\url}[1]{\texttt{#1}}
\expandafter\ifx\csname urlstyle\endcsname\relax
  \providecommand{\doi}[1]{doi: #1}\else
  \providecommand{\doi}{doi: \begingroup \urlstyle{rm}\Url}\fi

\bibitem[Albrecher et~al.(2006)Albrecher, Mayer, Schoutens, and
  Tistaert]{albrecher2006little}
H.~Albrecher, P.~Mayer, W.~Schoutens, and J.~Tistaert.
\newblock The little {H}eston trap.
\newblock \emph{KU Leuven Section of Statistics Technical Report}, 06\penalty0
  (05), 2006.

\bibitem[Bormetti et~al.(2016)Bormetti, Callegaro, Livieri, and
  Pallavicini]{bormetti2016backward}
G.~Bormetti, G.~Callegaro, G.~Livieri, and A.~Pallavicini.
\newblock A backward {M}onte {C}arlo approach to exotic option pricing.
\newblock 2016.
\newblock Available at SSRN 2686115.

\bibitem[Callegaro et~al.(2014)Callegaro, Fiorin, and
  Grasselli]{callegaro2014pricing}
G.~Callegaro, L.~Fiorin, and M.~Grasselli.
\newblock Pricing and calibration in local volatility models via fast
  quantization.
\newblock Available at SSRN 2495829, 2014.

\bibitem[Callegaro et~al.(2015{\natexlab{a}})Callegaro, Fiorin, and
  Grasselli]{Callegaroetal2015a}
G.~Callegaro, L.~Fiorin, and M.~Grasselli.
\newblock Quantized calibration in local volatility.
\newblock \emph{Risk Magazine}, 28:\penalty0 62--67, 2015{\natexlab{a}}.

\bibitem[Callegaro et~al.(2015{\natexlab{b}})Callegaro, Fiorin, and
  Grasselli]{callegaro2015pricing}
G.~Callegaro, L.~Fiorin, and M.~Grasselli.
\newblock Pricing via quantization in stochastic volatility models.
\newblock Available at SSRN 2669734, 2015{\natexlab{b}}.

\bibitem[Chen et~al.(2012)Chen, Oosterlee, and van~der
  Weide]{chen2012efficient}
B.~Chen, C.~Oosterlee, and J.~van~der Weide.
\newblock Efficient unbiased simulation scheme for {SABR} stochastic volatility
  model.
\newblock \emph{International Journal of Theoretical and Applied Finance},
  15\penalty0 (2), 2012.

\bibitem[Hagan et~al.(2002)Hagan, Kumar, Lesniewski, and
  Woodward]{hagan2002managing}
P.~S. Hagan, D.~Kumar, A.~S. Lesniewski, and D.~E. Woodward.
\newblock Managing smile risk.
\newblock \emph{The Best of Wilmott}, 1:\penalty0 249--296, 2002.

\bibitem[Hagan et~al.(2016)Hagan, Kumar, Lesniewski, and
  Woodward]{hagan2016universal}
P.~S. Hagan, D.~Kumar, A.~S. Lesniewski, and D.~E. Woodward.
\newblock Universal smiles.
\newblock \emph{Wilmott}, 2016\penalty0 (84):\penalty0 40--55, 2016.

\bibitem[Heston(1993)]{heston1993closed}
S.~L. Heston.
\newblock A closed-form solution for options with stochastic volatility with
  applications to bond and currency options.
\newblock \emph{Review of {F}inancial {S}tudies}, 6\penalty0 (2):\penalty0
  327--343, 1993.

\bibitem[Korn and Tang(2013)]{korn2013exact}
R.~Korn and S.~Tang.
\newblock Exact analytical solution for the normal {SABR} model.
\newblock \emph{Wilmott}, 2013\penalty0 (66):\penalty0 64--69, 2013.

\bibitem[Lord et~al.(2010)Lord, Koekkoek, and Dijk]{lord2010comparison}
R.~Lord, R.~Koekkoek, and D.~V. Dijk.
\newblock A comparison of biased simulation schemes for stochastic volatility
  models.
\newblock \emph{Quantitative Finance}, 10\penalty0 (2):\penalty0 177--194,
  2010.

\bibitem[McWalter et~al.(2017)McWalter, Rudd, Kienitz, and
  Platen]{mcwalter2017recursive}
T.~A. McWalter, R.~Rudd, J.~Kienitz, and E.~Platen.
\newblock Recursive marginal quantization of higher-order schemes.
\newblock Available at SSRN 2894753, 2017.

\bibitem[Obl{\'o}j(2007)]{obloj2007fine}
J.~Obl{\'o}j.
\newblock Fine-tune your smile: Correction to {H}agan et al.
\newblock \emph{arXiv:0708.0998}, 2007.

\bibitem[Pag{\`e}s(2014)]{pagesintroduction}
G.~Pag{\`e}s.
\newblock {Introduction to optimal vector quantization and its applications for
  numerics}.
\newblock Technical report, July 2014.
\newblock URL \url{https://hal.archives-ouvertes.fr/hal-01034196}.

\bibitem[{Pag{\`e}s} and Pham(2005)]{PagesPham2005}
G.~{Pag{\`e}s} and H.~Pham.
\newblock Optimal quantization methods for nonlinear filtering with
  discrete-time observations.
\newblock \emph{Bernoulli}, 11\penalty0 (5):\penalty0 893--932, 2005.

\bibitem[Pag{\`e}s and Sagna(2015)]{pages2015recursive}
G.~Pag{\`e}s and A.~Sagna.
\newblock Recursive marginal quantization of the {E}uler scheme of a diffusion
  process.
\newblock \emph{Applied Mathematical Finance}, 22\penalty0 (5):\penalty0
  463--498, 2015.

\bibitem[{Pag{\`e}s} and Wilbertz(2009)]{PagesWilbertz2012}
G.~{Pag{\`e}s} and B.~Wilbertz.
\newblock Optimal {Delaunay} and {Voronoi} quantization methods for pricing
  {American} options.
\newblock In R.~Carmona, P.~Hu, P.~Del~Moral, and N.~Oudjane, editors,
  \emph{Numerical {M}ethods in Finance}, pages 171--217. Springer, 2009.

\bibitem[{Pag{\`e}s} et~al.(2004){Pag{\`e}s}, Pham, and
  Printems]{PagesPhamPrintems2004}
G.~{Pag{\`e}s}, H.~Pham, and J.~Printems.
\newblock An optimal {M}arkovian quantization algorithm for multidimensional
  stochastic control problems.
\newblock \emph{Stochastics and Dynamics}, 4\penalty0 (4):\penalty0 501--545,
  2004.

\bibitem[Paulot(2015)]{paulot2015asymptotic}
L.~Paulot.
\newblock Asymptotic implied volatility at the second order with application to
  the {SABR} model.
\newblock In \emph{Large Deviations and Asymptotic Methods in Finance}, pages
  37--69. Springer, 2015.

\bibitem[Sagna(2011)]{sagna2010pricing}
A.~Sagna.
\newblock Pricing of barrier options by marginal functional quantization.
\newblock \emph{Monte Carlo Methods and Applications}, 17\penalty0
  (4):\penalty0 371--398, 2011.

\bibitem[Sch{\"o}bel and Zhu(1999)]{schobel1999stochastic}
R.~Sch{\"o}bel and J.~Zhu.
\newblock Stochastic volatility with an {O}rnstein--{U}hlenbeck process: an
  extension.
\newblock \emph{European Finance Review}, 3\penalty0 (1):\penalty0 23--46,
  1999.

\bibitem[Stein and Stein(1991)]{stein1991stock}
E.~M. Stein and J.~C. Stein.
\newblock Stock price distributions with stochastic volatility: an analytic
  approach.
\newblock \emph{Review of {F}inancial {S}tudies}, 4\penalty0 (4):\penalty0
  727--752, 1991.

\end{thebibliography}
